\def\year{2023}\relax
\documentclass[letterpaper]{article} 
\pdfoutput=1

\usepackage{aaai23}  
\usepackage{times}  
\usepackage{helvet}  
\usepackage{courier}  
\usepackage[hyphens]{url}  
\usepackage{graphicx} 
\usepackage{natbib}  
\usepackage{caption} 
\DeclareCaptionStyle{ruled}{labelfont=normalfont,labelsep=colon,strut=off} 
\usepackage{algorithm}
\usepackage{algorithmic}

\usepackage{newfloat}
\usepackage{listings}
\floatstyle{ruled}
\newfloat{listing}{tb}{lst}{}
\floatname{listing}{Listing}
\usepackage{booktabs}       
\usepackage{amsfonts}       
\usepackage{nicefrac}       
\usepackage{microtype}      
\usepackage[usenames]{color}
\usepackage[dvipsnames]{xcolor}         

\usepackage{amsthm}
\newtheorem{theorem}{Theorem}
\newtheorem{assumption}{Assumption}
\newtheorem{definition}{Definition}
\newtheorem{example}{Example}
\newtheorem{lemma}{Lemma}

\newtheorem*{problem*}{Formal problem}

\usepackage{amsmath,amssymb}       
\usepackage{cases}
\usepackage{empheq}
\usepackage{mathbbol}
\usepackage{bm}
\usepackage{bbm}
\DeclareSymbolFontAlphabet{\amsmathbb}{AMSb}
\usepackage{mathtools}
\usepackage{accents}
\usepackage[colorinlistoftodos]{todonotes}

\usepackage{siunitx}
\usepackage{caption}
\usepackage{subcaption}

\graphicspath{./Figures/}

\newcommand\munderbar[1]{%
  \underaccent{\bar}{#1}}

\usepackage{cleveref}

\crefname{equation}{Eq.}{Eqs.}
\crefname{pluralequation}{Eqs.}{Eqs.}

\crefname{algorithm}{Algorithm}{Algorithm}

\crefname{figure}{Fig.}{Figs.}
\crefname{pluralfigure}{Figs.}{Figs.}

\crefname{section}{Sect.}{Sects.}
\crefname{pluralsection}{Sects.}{Sects.}

\crefname{table}{Table}{Table}
\crefname{pluraltable}{Tables}{Tables}

\crefname{definition}{Def.}{Def.}
\crefname{pluraldefinition}{Defs.}{Defs.}

\crefname{theorem}{Theorem}{Theorems}
\crefname{pluraltheorem}{Theorems}{Theorems}

\crefname{lemma}{Lemma}{Lemmas}
\crefname{plurallemma}{Lemmas}{Lemmas}

\crefname{example}{Example}{Example}
\crefname{pluralexample}{Examples}{Examples}

\crefname{problem}{Problem}{Problem}
\crefname{pluralproblem}{Problems}{Problems}

\crefname{assumption}{Assumption}{Assumption}
\crefname{pluralassumption}{Assumptions}{Assumptions}

\crefname{remark}{Remark}{Remark}
\crefname{pluralremark}{Remarks}{Remarks}

\crefname{appendix}{Appendix}{Appendices}
\crefname{pluralappendix}{Appendices}{Appendices}

\usepackage{tikz} 
\usepackage{pgfplots}
\usepgfplotslibrary{external,fillbetween}
\pgfplotsset{compat=1.8}

\definecolor{red}{rgb}{0.745,0.192,0.102}
\definecolor{darkgreen}{RGB}{34,161,55}
\definecolor{ruhuisstijlrood}{rgb}{0.745,0.192,0.102}
\definecolor{ruhuisstijlzwart}{rgb}{0,0,0}
\definecolor{ruhuisstijlwit}{rgb}{0.98,0.98,0.98}
\newcommand{\rured}[1]{\textcolor{ruhuisstijlrood}{#1}}
\newcommand{\darkgreen}[1]{\textcolor{darkgreen}{#1}}

\usepackage{mdframed}

\usetikzlibrary{patterns,arrows,arrows.meta,calc,shapes,shadows,decorations.pathmorphing,decorations.pathreplacing,automata,shapes.multipart,positioning,shapes.geometric,fit,circuits,trees,shapes.gates.logic.US,fit, shadows.blur,shapes.symbols}

\usepackage[acronym]{glossaries}

\usepackage{fontawesome5}

\newcommand{\stochy}{\textrm{StocHy}}
\newcommand{\sreachtools}{\textrm{SReachTools}}
\newcommand{\probreach}{\textrm{ProbReach}}

\DeclareMathOperator*{\minimize}{minimize}

\DeclareMathOperator*{\argmax}{arg\,max}
\DeclareMathOperator*{\argmin}{arg\,min}

\newcommand{\R}{\amsmathbb{R}}

\newcommand{\N}{\amsmathbb{N}}
\newcommand{\distr}[1]{\ensuremath{\mathit{Dist(#1)}}}

\newcommand{\States}{\ensuremath{S}}
\newcommand{\Actions}{Act}

\newcommand{\initStateMDP}{\ensuremath{s_I}}
\newcommand{\transfunc}{\ensuremath{P}}
\newcommand{\transfuncImdp}{\ensuremath{\mathcal{P}}}
\newcommand{\transfuncLow}{\ensuremath{\underline{P}}}

\newcommand{\Interval}{\ensuremath{\amsmathbb{I}}}

\newcommand{\iMDP}{\ensuremath{(\States,\Actions,\initStateMDP,\transfuncImdp)}}

\newcommand{\imdp}{\ensuremath{\mathcal{M}_\Interval}}

\newcommand{\policy}{\ensuremath{\pi}}
\newcommand{\policySpace}{\ensuremath{\Pi}}

\newcommand{\controller}{\ensuremath{c}}

\newcommand{\cControlSpace}{\ensuremath{\mathcal{U}}}

\newcommand{\partitionSpace}{\ensuremath{\mathcal{X}}}
\newcommand{\region}{\ensuremath{\mathcal{P}}}
\newcommand{\boxSample}{\ensuremath{\mathcal{H}}}
\newcommand{\targetSet}{\ensuremath{\mathcal{T}}}
\newcommand{\discardedSet}{\ensuremath{R}}

\newcommand{\state}{\ensuremath{x}}
\newcommand{\control}{\ensuremath{u}}
\newcommand{\noise}{\ensuremath{\eta}}
\newcommand{\disturbance}{\ensuremath{q}}

\newcommand{\horizon}{\ensuremath{K}}
\newcommand{\Map}{\ensuremath{T}}

\newcommand{\threshold}{\ensuremath{\lambda}}

\newcommand{\SafeSet}{\ensuremath{\mathcal{Z}}}
\newcommand{\ReachSet}{\ensuremath{\mathcal{G}}}

\newcommand{\BackReach}{\ensuremath{\mathcal{R}^{-1}_{\hat\alpha}}}

\newcommand{\propertyMDP}{\ensuremath{\varphi_{\initStateMDP}^\horizon}}

\newcommand{\reachProbDiscr}{\ensuremath{Pr^\policy(\imdp[\transfunc] \models \propertyMDP)}}

\newcommand{\reachProbDiscrRobustStar}{\ensuremath{\underline{Pr}^{\policy^\star}(\imdp \models \propertyMDP)}}

\newcommand{\Gauss[2]}{\ensuremath{\mathcal{N}(#1 , #2)}}

\glsdisablehyper

\newacronym[]{LTL}{LTL}{linear temporal logic}
\newacronym[]{iid}{i.i.d.}{independent and identically distributed}
\newacronym[]{UAV}{UAV}{unmanned aerial vehicle}

\newacronym[]{PAC}{PAC}{probably approximately correct}
\newacronym[]{DRO}{DRO}{distributionally robust optimization}
\newacronym[plural=MDPs,firstplural=Markov decision processes (MDPs)]{MDP}{MDP}{Markov decision process}
\newacronym[plural=iMDPs,firstplural=interval Markov decision processes (iMDPs)]{iMDP}{iMDP}{interval Markov decision process}
\newacronym[plural=aMDPs,firstplural=augmented MDPs (aMDPs)]{aMDP}{aMDP}{augmented MDP}
\newacronym[plural=POMDPs,firstplural=partially observable Markov decision processes (POMDPs)]{POMDP}{POMDP}{partially observable Markov decision process}

\glsunset{iMDP}
\glsunset{iid}

\newcommand{\change}[1]{#1}

\newif\ifappendix
\global\appendixtrue

\title{Probabilities Are Not Enough: Formal Controller Synthesis for \\ Stochastic Dynamical Models with Epistemic Uncertainty}
\author{
     Thom Badings\textsuperscript{\rm 1}, 
     Licio Romao\textsuperscript{\rm 2}, 
     Alessandro Abate\textsuperscript{\rm 2}, 
     Nils Jansen\textsuperscript{\rm 1}
}
\affiliations {
     \textsuperscript{\rm 1} Radboud University, Nijmegen, the Netherlands\\
     \textsuperscript{\rm 2} University of Oxford, Oxford, United Kingdom\\
}

\definecolor{plotblue}{rgb}{0.1,0.498039215686275,0.9549019607843137}

\begin{document}

\maketitle

\begin{abstract}
Capturing uncertainty in models of complex dynamical systems is crucial to designing safe controllers. Stochastic noise causes aleatoric uncertainty, whereas imprecise knowledge of model parameters leads to epistemic uncertainty. Several approaches use formal abstractions to synthesize policies that satisfy temporal specifications related to safety and reachability. However, the underlying models exclusively capture aleatoric but not epistemic uncertainty, and thus require that model parameters are known precisely. Our contribution to overcoming this restriction is a novel abstraction-based controller synthesis method for continuous-state models with stochastic noise and uncertain parameters. By sampling techniques and robust analysis, we capture both aleatoric and epistemic uncertainty, with a user-specified confidence level, in the transition probability intervals of a so-called interval Markov decision process (iMDP). We synthesize an optimal policy on this iMDP, which translates (with the specified confidence level) to a feedback controller for the continuous model with the same performance guarantees. Our experimental benchmarks confirm that accounting for epistemic uncertainty leads to controllers that are more robust against variations in parameter values.
\end{abstract}

\section{Introduction}
\label{sec:Introduction}

\paragraph{Stochastic models.}
Stochastic dynamical models capture complex systems where the likelihood of transitions is specified by probabilities~\cite{kumar2015stochastic}. 
Controllers for stochastic models must act safely and reliably with respect to a desired specification. 
Traditional control design methods use, e.g., Lyapunov functions and optimization to guarantee stability and (asymptotic) convergence.
However, alternative methods are needed to give formal guarantees about richer temporal specifications relevant to, for example, safety-critical applications~\cite{DBLP:journals/tac/FanQMNMV22}.

\paragraph{Finite abstractions.}
A powerful approach to synthesizing certifiably safe controllers leverages probabilistic verification to provide formal guarantees over specifications of \change{\emph{safety} (always avoid certain states) and \emph{reachability} (reach a certain set of states).}
A common example is the \emph{reach-avoid} specification, where the task is to maximize the probability of reaching desired goal states while avoiding unsafe  states~\cite{DBLP:conf/hybrid/FisacCTS15}.
\change{Finite abstractions can make continuous models amenable to techniques and tools from formal verification: by discretizing their state and action spaces, abstractions result in, e.g., finite Markov decision processes (MDPs) that soundly capture the continuous dynamics~\cite{Abate2008probabilisticSystems}.}
Verification guarantees on the finite abstraction can thus carry over to the continuous model.
In this paper, we adopt such an abstraction-based approach to controller synthesis.

\paragraph{Probabilities are not enough.}
The notion of uncertainty is often distinguished in \emph{aleatoric} (statistical) and \emph{epistemic} (systematic)  uncertainty~\cite{fox2011distinguishing,sullivan2015introduction}.  
Epistemic uncertainty is, in particular, modeled by parameters that are not precisely known~\cite{smith2013uncertainty}. 
A general premise is that purely probabilistic approaches fail to capture epistemic uncertainty~\cite{DBLP:journals/ml/HullermeierW21}. 
In this work, we aim to reason under both aleatoric and epistemic uncertainty in order to synthesize provably correct controllers for safety-critical applications. 
Existing abstraction methods fail to achieve this novel, general goal.

\begin{figure}[t!]
	\centering
	\includegraphics[width = \columnwidth]{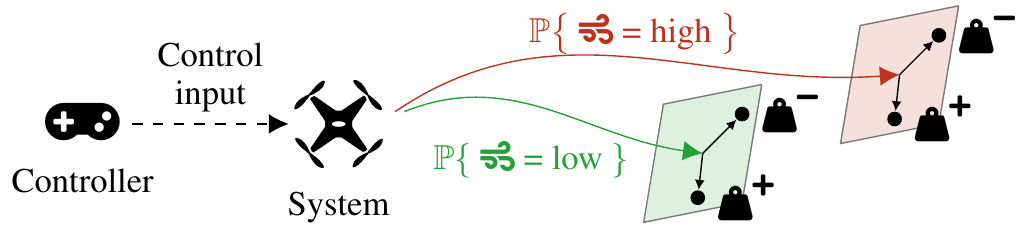}
	\caption{Aleatoric (stochastic) uncertainty in the wind (\faIcon{wind}) causes probability distributions over the outcomes of controls, while epistemic uncertainty in the mass (\faIcon{weight-hanging}) of the drone causes state transitions to be nondeterministic.}
	\label{fig:nondeterminism_vs_probabilities}
\end{figure}

\paragraph{Models with epistemic uncertainty.}
\change{
We consider reach-avoid problems for stochastic dynamical models with continuous state and action spaces under epistemic uncertainty described by parameters that lie within a \emph{convex uncertainty set}.
In the simplest case, this uncertainty set is an interval, such as a drone whose mass is only known to lie between $0.75$--$\SI{1.25}{\kilogram}$.
As shown in \cref{fig:nondeterminism_vs_probabilities}, the dynamics of the drone depend on uncertain factors, such as the wind and the drone's mass.
For the wind, we may derive a probabilistic model from, e.g., weather data to reason over the likelihood of state dynamics. 
For the mass, however, we do not have information about the likelihood of each value, so employing a probabilistic model is unrealistic.
Thus, we treat epistemic uncertainty in such imprecisely known parameters (in this case, the mass) using a \emph{nondeterministic framework} instead.}

\paragraph{Problem statement.}
Our goal is to synthesize a controller that (1) is \emph{robust against nondeterminism} due to parameter uncertainty and (2) \emph{reasons over probabilities} derived from stochastic noise.
\change{In other words, the controller must satisfy a given specification \emph{under any possible outcome of the nondeterminism} (robustness) and \emph{with at least a certain probability regarding the stochastic noise} (reasoning over probabilities).}
We wish to synthesize a controller with a \emph{\gls{PAC}}-style guarantee to satisfy a reach-avoid specification with at least a desired threshold probability.
Thus, we solve the following problem:

\begin{mdframed}[backgroundcolor=red!10, nobreak=true, innerleftmargin=8pt, innerrightmargin=8pt]
\textbf{Problem.} Given a reach-avoid specification for a stochastic model with uncertain parameters, compute a controller and a \emph{lower bound} on the probability that, \emph{under any admissible value of the parameters}, the specification is probabilistically satisfied with this lower bound and \emph{with at least a user-specified confidence probability}.
\end{mdframed}

\noindent
We solve this problem via a discrete-state abstraction of the continuous model.
We generate this abstraction by partitioning the continuous state space and defining actions that induce potential transitions between elements of this partition. 

Algorithmically, the closest approach to ours is~\citet{Badings2022AAAI}, which uses abstractions to synthesize controllers for stochastic models with aleatoric uncertainty of unknown distribution, but with \emph{known parameters}.
Our setting is more general, as epistemic uncertainty requires fundamental differences to the technical approach, as~explained~below. 

\paragraph{Robustness to capture nondeterminism.}
The main contribution that allows us to capture nondeterminism, is that we reason over \emph{sets} of potential transitions (as shown by the boxes in \cref{fig:nondeterminism_vs_probabilities}), rather than \emph{precise} transitions, e.g., as in~\citet{Badings2022AAAI}.
Intuitively, for a given action, the aleatoric uncertainty creates a probability distribution over \emph{sets of possible outcomes}. 
To ensure robustness against epistemic uncertainty, we consider \emph{all possible outcomes} within these sets.
We show that, for our class of models, computing these sets of all possible outcomes is computationally tractable. 
Building upon this reasoning, we provide the following guarantees to solve the above-mentioned problem.

\paragraph{1) PAC guarantees on abstractions.}
We show that both probabilities and nondeterminism can be captured in the transition probability intervals of so-called interval Markov decision processes (iMDPs,~\citealt{DBLP:journals/ai/GivanLD00}).
We use sampling methods from scenario optimization~\cite{DBLP:journals/arc/CampiCG21} and concentration inequalities~\cite{Boucheron2013concentrationInequalities} to compute \gls{PAC} bounds on these intervals. 
With a predefined confidence probability, the \gls{iMDP} correctly captures both aleatoric and epistemic~uncertainty.

\paragraph{2) Correct-by-construction control.} 
For the \gls{iMDP}, we compute a \emph{robust optimal policy} that maximizes the worst-case probability of satisfying the reach-avoid specification.
The \gls{iMDP} policy is automatically translated to a controller for the original, continuous model on the fly.
We show that, by construction, the \gls{PAC} guarantees on the \gls{iMDP} carry over to the satisfaction of the specification by the continuous model.

\paragraph{Contributions.}
We develop the first abstraction-based, formal controller synthesis method that simultaneously captures epistemic and aleatoric uncertainty for models with continuous state and action spaces. 
We also provide results on the \gls{PAC}-correctness of obtained \gls{iMDP} abstractions and guarantees on the synthesized controllers for a reach-avoid specification. 
Our numerical experiments in \cref{sec:Experiments} confirm that accounting for epistemic uncertainty yields controllers that are more robust against deviations in the parameter values.
\subsection*{Related Work}

\paragraph{Uncertainty models.}
Distinguishing aleatoric from epistemic uncertainty is a key challenge towards trustworthy AI~\cite{DBLP:journals/electronicmarkets/ThiebesLS21}, and has been considered in reinforcement learning~\cite{DBLP:Charpentier:journals/corr/abs-2206-01558}, Bayesian neural networks~\cite{DBLP:conf/icml/DepewegHDU18,DBLP:journals/ral/LoquercioSS20}, 
and systems modeling~\cite{smith2013uncertainty}.
Dynamical models with epistemic parameter uncertainty (but \emph{without aleatoric uncertainty}) are considered by~\citet{yedavalli2014robust}, and~\citet{geromel2006robust}.
Control of models similar to ours is studied by~\cite{DBLP:Modares:journals/corr/abs-2202-04495}, albeit only for stability specifications.

\paragraph{Model-based approaches.}
Abstractions of stochastic models are a well-studied research area~\cite{Abate2008probabilisticSystems,Alur2000discretesystems}, with applications to stochastic hybrid~\cite{DBLP:conf/hybrid/CauchiLLAKC19,LavaeiSAZ21:stochhybrid}, switched~\cite{DBLP:journals/tac/LahijanianAB15}, and partially observable systems~\cite{Badings2021FilterUncertainty,DBLP:conf/adhs/HaesaertNVTAAM18}.
Various tools exist, e.g., \stochy~\citep{DBLP:conf/tacas/CauchiA19}, \probreach~\citep{DBLP:conf/hybrid/ShmarovZ15}, and \sreachtools~\citep{DBLP:conf/hybrid/VinodGO19}.
However, in contrast to the approach of this paper, \emph{none of these papers deals with epistemic uncertainty}.

\citet{DBLP:journals/tac/FanQMNMV22} use optimization for reach-avoid control of linear but \emph{non-stochastic} models with bounded disturbances.
\change{Barrier functions are used for cost minimization in stochastic optimal control~\cite{DBLP:conf/corl/Pereira0ET20}.
So-called \emph{funnel libraries} are used by~\citet{DBLP:journals/ijrr/MajumdarT17} for robust feedback motion planning under~epistemic~uncertainty.
Finally, \citet{Zikelic2022} learn policies together with formal reach-avoid certificates using neural networks for nonlinear systems with only aleatoric uncertainty.}

\paragraph{Data-driven approaches.}

Models with (partly) unknown dynamics express epistemic uncertainty about the underlying system.
Verification of such models based on data has been done using Bayesian inference~\cite{DBLP:journals/automatica/HaesaertHA17}, optimization~\cite{DBLP:journals/automatica/KenanianBJT19,DBLP:Vinod:journals/corr/abs-2206-11103}, and Gaussian process regression~\cite{DBLP:conf/cdc/JacksonLFL20}.
\change{Moreover,~\citet{DBLP:journals/ral/KnuthCOB21}, and~\citet{DBLP:conf/cdc/ChouOB21} use neural network models for feedback motion planning for nonlinear deterministic systems with probabilistic safety and reachability guarantees.}
Data-driven abstractions have been developed for monotone~\cite{DBLP:conf/l4dc/MakdesiGF21} and event-triggered systems~\cite{DBLP:journals/csysl/PeruffoM23}.
By contrast to our setting, \emph{these approaches consider models with non-stochastic dynamics}.
A few recent exceptions also consider aleatoric uncertainty~\cite{salamati2022safety,lavaei2022constructing}, but these approaches require more strict assumptions (e.g., discrete input sets) than our model-based approach.

\paragraph{Safe learning.}
While outside the scope of this paper, our approach fits naturally in a model-based safe learning context~\cite{Brunke2021safelearning,DBLP:journals/jmlr/GarciaF15}.
In such a setting, our approach may synthesize controllers that guarantee safe interactions with the system, while techniques from, for example, reinforcement learning (RL, \citealt{DBLP:conf/nips/BerkenkampTS017,DBLP:journals/tac/ZanonG21}) or stochastic system identification~\cite{DBLP:conf/cdc/TsiamisP19} can reduce the epistemic uncertainty based on state observations.
\change{A risk-sensitive RL scheme providing \emph{approximate} safety guarantees is developed by \citet{DBLP:journals/jair/GeibelW05}; we instead give \emph{formal} guarantees at the cost of an expensive abstraction.}
\section{Problem Statement}
\label{sec:Problem}

The cardinality of a discrete set $\mathcal{X}$ is $|\mathcal{X}|$. 
A probability space is a triple $(\Omega, \mathcal{F}, \amsmathbb{P})$ of an arbitrary set $\Omega$, sigma algebra $\mathcal{F}$ on $\Omega$, and probability measure $\amsmathbb{P} \colon \mathcal{F} \to [0,1]$. 
The convex hull of a polytopic set $\mathcal{X}$ with vertices $v_1, \ldots, v_n$ is $\mathrm{conv} (v_1, \ldots, v_n)$.
The word \emph{controller} relates to continuous models; a \emph{policy} to discrete models.

\subsection*{Stochastic Models with Parameter Uncertainty}
To capture parameter uncertainty in a linear time-invariant stochastic system, we consider a model \change{(we extend this model with parameters describing uncertain additive disturbances in \cref{sec:Algorithm})} whose continuous state $\state_k$ at time $k \in \N$ evolves as
\begin{equation}
    \state_{k+1} = A(\alpha) \state_k + B(\alpha) \control_k + \noise_k,
    \label{eq:uncertain_LTI_model}
\end{equation}
where $\control_k \in \cControlSpace$ is the control input, which is constrained by the control space $\cControlSpace = \mathrm{conv} (u^1, \ldots, u^q) \subset \R^m$, being is a convex polytope with $q$ vertices, and where the process noise $\noise_k$ is a stochastic process defined on a probability space $(\Omega, \mathcal{F}, \amsmathbb{P})$.
Both the dynamics matrix $A(\alpha) \in \R^{n \times n}$ and the control matrix $B(\alpha) \in \amsmathbb{R}^{n \times m}$ are a convex combination of a finite set of $r \in \N$ known matrices:
\begin{align}
    A(\alpha) &= \sum\nolimits_{i = 1}^r \alpha_i A_i, \enskip
    B(\alpha) = \sum\nolimits_{i = 1}^r \alpha_i B_i, \enskip
    \label{eq:dynamics_matrix_vertices}
\end{align}
where the \emph{unknown model parameter} $\alpha \in \Gamma$ can be any point in the unit simplex $\Gamma \subset \R^r$:
\begin{align*}
\Gamma
    &= \Big\{ 
        \alpha \in \amsmathbb{R}^r \colon \alpha_i \geq 0,~\forall i \in \{1, \ldots, r\}, \enskip 
        \sum\nolimits_{i = 1}^r \alpha_i = 1 
    \Big\}.
    \label{eq:Simplex_def} 
\end{align*}
The model in \cref{eq:uncertain_LTI_model} captures epistemic uncertainty in $A(\alpha)$ and $B(\alpha)$ through model parameter $\alpha$,
as well as aleatoric uncertainty in the discrete-time stochastic process $(\noise_k)_{k \in \amsmathbb{N}}$.
\begin{assumption}
    The noise $\noise_k$ is independent and identically distributed (i.i.d., which is a common assumption) and has density with respect to the Lebesgue measure.
    However, contrary to most definitions, we allow $\amsmathbb{P}$ to be \emph{unknown}.
\label{assump:iid_process}
\end{assumption}
\noindent 
\change{Importantly, being distribution-free, our proposed techniques hold for any distribution of $\eta$ that satisfies \cref{assump:iid_process}.}

The matrices $A_i$ and $B_i$ can represent the bounds of intervals over parameters, as illustrated by the following example.
\begin{example}
    \label{ex:Model}
    Consider again the drone of \cref{fig:nondeterminism_vs_probabilities}.
    The drone's longitudinal position $p_k$ and velocity $v_k$ are modeled as
    \begin{equation*}
        \state_{k+1} = 
        \begin{bmatrix}
            p_{k+1} \\ v_{k+1}
        \end{bmatrix} = 
        \begin{bmatrix}
            1 & \tau \\ 0 & 1 - \frac{0.1 \tau}{m}
        \end{bmatrix} \state_k +
        \begin{bmatrix}
            \frac{\tau^2}{2m} \\ \frac{\tau}{m}
        \end{bmatrix} \control_k + 
        \noise_k,
    \end{equation*}
    with $\tau$ the discretization time, and $\cControlSpace = [-5, 5]$.
    Assume that the mass $m$ is only known to lie within $[0.75, 1.25]$.
    Then, we obtain a model as \cref{eq:uncertain_LTI_model}, with $r=2$ vertices where $A_1, B_1$ are obtained for $m \coloneqq 0.75$, and $A_2, B_2$ for $m \coloneqq 1.25$.
    \qed
\end{example}

\subsection*{Reach-Avoid Planning Problem}

The goal is to steer the state $\state_k$ of \cref{eq:uncertain_LTI_model} to a desirable state within $K$ time steps while always remaining in a safe region.
Formally, let the \emph{safe set} $\SafeSet$ be a compact set of $\amsmathbb{R}^n$, and let $\ReachSet \subseteq \SafeSet$ be the \emph{goal set} (see \cref{fig:Partition}).
The control inputs $\control_k$ in \cref{eq:uncertain_LTI_model} are selected by a time-varying \emph{feedback controller}:
\begin{definition}
    A time-varying feedback \emph{controller} $\controller \colon \R^n \times \N \to \cControlSpace$ for \cref{eq:uncertain_LTI_model} is a function that maps a state $\state_k \in \R^n$ and a time step $k \in \N$ to a control input $\control_k \in \cControlSpace$.
    \label{def:ControlLaw}
\end{definition}
\noindent
The space of admissible feedback controllers is denoted by $\mathcal{C}$.
The \emph{reach-avoid probability} $V(\state_0, \alpha, c)$ is the probability that \cref{eq:uncertain_LTI_model}, under parameter $\alpha \in \Gamma$ and a controller $\controller$, satisfies a reach-avoid planning problem with respect to the sets $\SafeSet$ and $\ReachSet$, starting from initial state $\state_0 \in \R^n$. Formally:
\begin{definition}
    The \emph{reach-avoid probability} $V \colon \R^n \times \Gamma \times \mathcal{C} \to [0,1]$ for a given controller $c \in \mathcal{C}$ on horizon $K$ is
    \begin{equation*}
    \begin{split}
    V(\state_0, &\ \alpha, c) = \amsmathbb{P} \big\{
        \noise_k \in \Omega
        \, \colon \,
        \state_{k} \text{ evolves as per \cref{eq:uncertain_LTI_model}},
        \\
        & \exists k' \in \{0,\ldots,K\} 
        \,\, \text{such that} \,\,
        \state_{k'} \in \ReachSet, \text{ and}
        \\
        & \state_k \in \SafeSet, \, \control_k = \controller(\state_k, k) \,\, \forall k \in \{0,\ldots,k' \}
        \big\}.
    \end{split}
    \end{equation*}
\end{definition}
\noindent
We aim to find a controller for which the reach-avoid probability is above a certain threshold $\lambda \in [0,1]$.
However, since parameter $\alpha$ of \cref{eq:uncertain_LTI_model} is unknown, it is impossible to compute the reach-avoid probability explicitly.  
We instead take a \emph{robust approach}, thus stating the formal problem as follows:
\begin{problem*}
    \label{prob:FormalProblem}
    \change{Given an initial state $\state_0$, find a controller $\controller \in \mathcal{C}$ together with a (high) probability threshold $\lambda \in [0,1]$, such that $V(\state_0, \alpha, c) \geq \lambda$ holds for all $\alpha \in \Gamma$.}
\end{problem*}
\noindent
Thus, we want to find a controller $\controller \in \mathcal{C}$ that is robust against \emph{any possible instance} of parameter $\alpha \in \Gamma$ in \cref{eq:uncertain_LTI_model}.
Due to the aleatoric uncertainty in \cref{eq:uncertain_LTI_model} of unknown distribution, we solve the problem up to a user-specified confidence probability $\beta \in (0,1)$, as we shall see in \cref{sec:Algorithm}.

\begin{figure}[!t]
    \centering
    \includegraphics[width = 0.65\columnwidth]{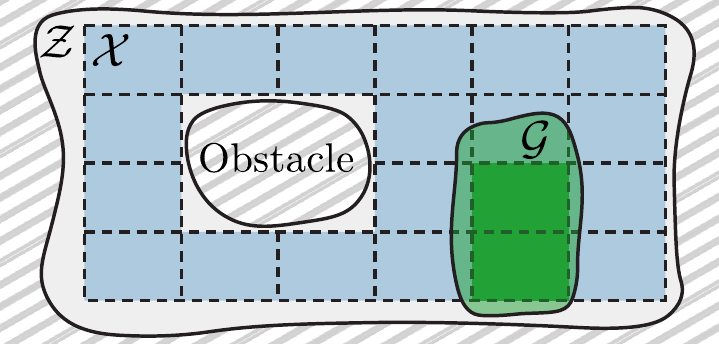}
    \caption{Partition of safe set $\partitionSpace \subseteq \SafeSet$ (which excludes obstacles) into regions that define \gls{iMDP} states.
    A state $s_i$ is a goal state, $s_i \in \States_\ReachSet$, if its region is contained in the goal region $\ReachSet$.}
    \label{fig:Partition}
\end{figure}

\subsection*{Interval Markov Decision Processes}

We solve the problem by generating a finite-state abstraction of the model as an \glsfirst{iMDP}:
\begin{definition}
    An \emph{\gls{iMDP}} is a tuple $\imdp=\iMDP$ where 
    $\States$ is a finite set of states, 
    $\Actions$ is a finite set of actions, 
    $\initStateMDP \in \States$ is the initial state, and $\transfuncImdp \colon \States \times \Actions \times \States \to \Interval \cup \{ [0,0] \}$ is an uncertain partial probabilistic transition function over intervals $\Interval = \{ [a,b] \ | \ a,b \in (0,1] \text{ and } a \leq b \}$.
    \label{def:iMDP}
\end{definition}
\noindent
Note that an interval may not have a lower bound of $0$, except for the $[0,0]$ interval.
If $\transfuncImdp(s,a,s') = [0,0] \, \forall s' \in \States$, action $a \in \Actions$ is not enabled in state $s \in \States$.
We can instantiate an \gls{iMDP} into an MDP with a partial transition function $\transfunc \colon \States \times \Actions \rightharpoonup \distr{\States}$ by fixing a probability $P(s,a)(s') \in \transfuncImdp(s,a,s')$ for each $s, s' \in \States$ and for each $a \in \Actions$ enabled in $s$, such that $P(s,a)$ is a probability distribution over $\States$.
For brevity, we write this instantiation as $P \in \transfuncImdp$ and denote the resulting MDP by $\imdp[P]$.

A deterministic policy~\cite{DBLP:books/daglib/BaierKatoen2008} for an \gls{iMDP} $\imdp$ is a function $\policy \colon \States^\ast \to \Actions$, where $\States^\ast$ is a sequence of states (memoryless policies do not suffice for time-bounded specifications), with $\policy \in \policySpace_{\imdp}$ the admissible policy space.
For a policy $\policy$ and an instantiated MDP $\imdp[P]$ for $P \in \transfuncImdp$, we denote by $\reachProbDiscr$ the probability of satisfying a reach-avoid specification\footnote{\change{Note that $\propertyMDP$ is a reach-avoid specification for an \gls{iMDP}, while $V(\state_0, \alpha, c)$ is the reach-avoid probability on the dynamical model.}} $\propertyMDP$ (i.e., reaching a goal in $\States_\ReachSet \subseteq \States$ within $K \in \N$ steps, while not leaving a safe set $\States_\SafeSet \subseteq \States$).
A \emph{robust optimal policy} $\policy^\star \in \policySpace_{\imdp}$ \emph{maximizes} this probability under the \emph{minimizing} instantiation $P \in \transfuncImdp$:\footnote{\change{Such \emph{min-max} (and \emph{max-min}) problems are common in (distributionally) robust optimization, see~\citet{DBLP:books/degruyter/Ben-TalGN09}, and~\citet{DBLP:journals/ior/WiesemannKS14} for details.}}
\begin{equation}
\begin{split}
    \label{eq:optimalPolicyRobust}
	\policy^\star 
	&= \argmax_{\policy \in \policySpace} \min_{\transfunc \in \transfuncImdp} \reachProbDiscr.
\end{split}
\end{equation}
We compute an optimal policy in \cref{eq:optimalPolicyRobust} \change{using a robust variant of value iteration proposed by~\citet{DBLP:conf/cdc/WolffTM12}.
Note that deterministic policies suffice to obtain optimal values for \cref{eq:optimalPolicyRobust}, see~\citet{DBLP:conf/cav/PuggelliLSS13}.}
\section{Finite-State Abstraction}
\label{sec:Abstractions}

To solve the formal problem, we construct a finite-state abstraction of \cref{eq:uncertain_LTI_model} as an \gls{iMDP}.
We define the actions of this \gls{iMDP} via backward reachability computations on a so-called \emph{nominal model} that neglects any source of uncertainty. 
We then compensate for the error caused by this modeling simplification in the \gls{iMDP}'s transition probability intervals.

\subsection*{Nominal Model of the Dynamics} 

To build our abstraction, we rely on a \emph{nominal model} that neglects both the aleatoric and epistemic uncertainty in \cref{eq:uncertain_LTI_model}, and is thus deterministic.
Concretely, we fix any value $\hat\alpha \in \Gamma$ and define the nominal model dynamics as
\begin{equation}
    \hat{\state}_{k+1} = A(\hat\alpha) \state_k + B(\hat\alpha) \control_k.
    \label{eq:nominal_model}
\end{equation}
Due to the linearity of the dynamics, we can now express the successor state $\state_{k+1}$ \emph{with full uncertainty}, from \cref{eq:uncertain_LTI_model}, as
\begin{equation}
    \state_{k+1} = \hat{\state}_{k+1} + \delta(\alpha,\state_k,\control_k) + \noise_k,
\label{eq:AbstractionError}
\end{equation}
with $\delta(\alpha,\state_k,\control_k)$ being a new term, called the \emph{epistemic error}, encompassing the error caused by parameter uncertainty:
\begin{equation}
    \delta(\alpha,\state_k,\control_k) = [A(\alpha) - A(\hat{\alpha})]\state_k + [B(\alpha) - B(\hat{\alpha})]\control_k. 
    \label{eq:EpistemicUncertainty_def}
\end{equation}
In other words, the successor state $\state_{k+1}$ is the nominal one, plus the epistemic error, and plus the stochastic noise.
Note that for $\alpha = \hat\alpha$ (i.e., the true model parameters equal their nominal values), we obtain $\delta(\alpha,\state_k,\control_k) = 0$. 
We also impose the next assumption on the nominal model, \change{which guarantees that we can compute the inverse image of \cref{eq:nominal_model} for a given $\hat{\state}_{k+1}$, 
\ifappendix
    and is used in the proof of \cref{lemma:BackReachSet_Representation} in \cref{sec:proofs}.
\else
    and is used in the proof of Lemma 3 in the extended version of this paper~\cite[Appendix~A]{Badings2022_AAAI_extended}.
\fi }%
\begin{assumption}
        \label{assump:nominal_nonsingular}
        The matrix $A(\hat \alpha)$ in \cref{eq:nominal_model} is non-singular.
\end{assumption}
\noindent
Assumption \cref{assump:nominal_nonsingular} is mild as it only requires the existence of a non-singular matrix in $\mathrm{conv}\{A_1, \ldots, A_r\}$.

\subsection*{States}

We create a partition of a subset $\partitionSpace$ of the safe set $\SafeSet$ on the continuous state space, see \cref{fig:Partition}.
This partition fully covers the goal set $\ReachSet$ but excludes unsafe states, i.e., any $\state \notin \SafeSet$.

\begin{definition}
    \label{def:Partition}
    A finite collection of subsets $(\region_i)_{i=1}^L$ is called a \emph{partition} of $\partitionSpace \subseteq \SafeSet \subset \R^n$ if the following conditions hold: 
    \begin{enumerate}
        \item $\partitionSpace = \bigcup_{i=1}^L \region_i$,
        \item $\region_i \bigcap \region_j = \emptyset, \,\, \forall i,j \in \{1,\ldots,L\}, \,\, i \neq j$.
    \end{enumerate}
\end{definition}

\noindent
We append to the partition a so-called \emph{absorbing region} $\region_0 = \mathrm{cl}(\R^n \setminus \partitionSpace)$, \change{which is defined as the closure of $\R^n \setminus \partitionSpace$} and represents any state $\state \notin \partitionSpace$ that is disregarded in subsequent reachability computations.
We consider partitions into convex polytopic regions, \change{which will allow us to compute \gls{PAC} probability intervals in \cref{lemma:prob_lower_bound} using results from~\citet{romao2022exact} on scenario optimization programs with discarded constraints}:

\begin{assumption}
    Each region $\region_i$ is a convex polytope given~by 
    \begin{align}
        \region_i &= \{ \state \in \amsmathbb{R}^n \colon H_i \state \leq h_i \}, 
        \label{eq:Polyhedron}
    \end{align}
    \change{with $H_i \in \amsmathbb{R}^{p_i \times n}$ and $h_i \in \amsmathbb{R}^{p_i}$ for some $p_i \in \N$,} and the inequality in \cref{eq:Polyhedron} is to be interpreted element-wise.
    \label{assump:Partitions}
\end{assumption}
\noindent
We define an \gls{iMDP} state for each element of $(\region_i)_{i=0}^L$, yielding a set of $L+1$ discrete states $\States = \{ s_i \mid i=0,\ldots,L \}$.
Define $\Map \colon \R^n \to \{0,1,\dots,L\}$ as the map from any $\state \in \partitionSpace$ to its corresponding region index $i$.
We say that a continuous state $\state$ belongs to \gls{iMDP} state $s_i$ if $\Map(\state) = i$.
State $s_0$ is a deadlock, such that the only transition leads back to $s_0$.

\subsection*{Actions}

Recall that we define the \gls{iMDP} actions via backward reachability computations under the nominal model in \cref{eq:nominal_model}.
Let $\targetSet = \{ \targetSet_1, \ldots, \targetSet_M \}$ be a finite collection of \emph{target sets}, each of which is a convex polytope, $\targetSet_\ell = \mathrm{conv}\{ t^1, \ldots, t^d \} \subset \R^n$.
Every target set corresponds to an \gls{iMDP} action, yielding the set $\Actions = \{ a_\ell \mid \ell=1,\ldots,M \}$ of actions.
Action $a_\ell \in \Actions$ represents a transition to $\hat{\state}_{k+1} \in \targetSet_\ell$ that is feasible under the nominal model.
The one-step \emph{backward reachable set} $\BackReach(\targetSet_\ell)$, shown in \cref{fig:backreachset}, represents precisely these continuous states from which such a direct transition to $\targetSet_\ell$ exists:
\begin{equation*}
    \BackReach(\targetSet_\ell)
	= \big\{ \state \in \R^n \, | \, \exists \control \in \cControlSpace, 
	\, A(\hat\alpha) \state + B(\hat\alpha) \control \in \targetSet_\ell
	\big\}.
\end{equation*}
Intuitively, we enable action $a_\ell$ in state $s_i$ only if $\hat{\state}_{k+1} \in \targetSet_\ell$ can be realized from any $\state_k \in \region_i$ in the associated region, or in other words, region $\region_i$ must be contained in the backward reachable set.
As such, we obtain the following definition:
\begin{definition}
    \label{def:Actions}
    Given a fixed $\hat\alpha \in \Gamma$, an action $a_\ell \in \Actions$ is enabled in a state $s_i \in \States$ if $\region_i \subseteq \BackReach(\targetSet_\ell)$. 
    The set $\Actions_{\hat\alpha}(s_i)$ of enabled actions in state $s_i$ under $\hat\alpha \in \Gamma$ is defined as
    \begin{equation}
        \Actions_{\hat\alpha}(s_i) = \big\{ a_\ell \in \Actions \colon \region_i \subseteq \BackReach(\targetSet_\ell) \big\}.
    \end{equation}
\end{definition}

\begin{figure}[t!]
	\centering
	\includegraphics[width = .85\columnwidth]{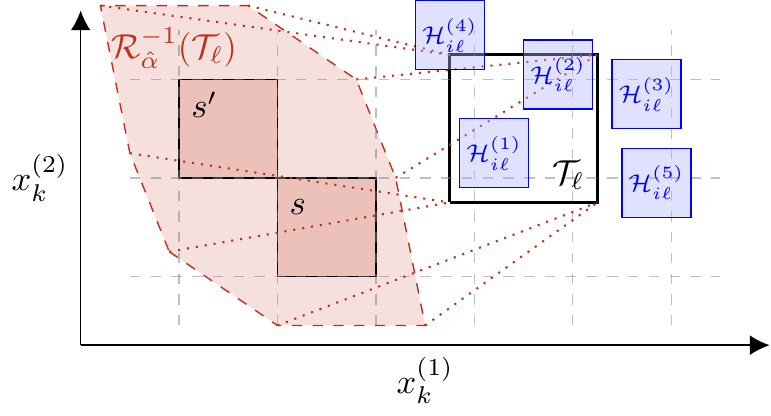}
	\caption{
	    Action $a_\ell$ to target set $\targetSet_\ell$ is enabled in states $s,s' \in \States$, as their regions are contained in the backward reachable set $\BackReach(\targetSet_\ell)$.
	    The successor state sets $\boxSample_{i\ell}^{(1)}, \ldots, \boxSample_{i\ell}^{(5)}$ for five noise samples, overapproximated as boxes, are shown in blue.
	}
	\label{fig:backreachset}
\end{figure}
\paragraph{Computing backward reachable sets.}
To apply \cref{def:Actions}, we must compute $\BackReach(\targetSet_\ell)$ for each action $a_\ell \in \Actions$ with associated target set $\targetSet_\ell$.
\change{
\ifappendix
    \cref{lemma:BackReachSet_Representation}, which is stated in \cref{sec:proofs} for brevity, shows
\else
    In the extended version of this paper~\cite[Appendix~A, Lemma 3]{Badings2022_AAAI_extended}, we show
\fi }%
that $\BackReach(\mathcal{T}_\ell)$ is a polytope characterized by the vertices of $\cControlSpace$ and $\targetSet_\ell$, which is computationally tractable to compute.

\subsection*{Transition Probability Intervals}

We wish to compute the probability $P(s_i,a_\ell)(s_j)$ that taking action $a_\ell \in \Actions$ in a continuous state $\state_k \in \region_i$ yields a successor state $\state_{k+1} \in \region_j$ that belongs to state $s_j \in \States$.
This conditional probability is defined using \cref{eq:AbstractionError} as
\begin{align}
    P(s_i, a_\ell)(s_j) 
    &= \amsmathbb{P} \{ \state_{k+1} \in \region_j \mid a_\ell \in \Actions_{\hat\alpha}(s_i) \}.
    \label{eq:transition_probability}
\end{align}
Note the outcome of an action $a_\ell$ is the same for any origin state $s_i$ in which $a_\ell$ is enabled.
In the remainder, we thus drop the conditioning on $a_\ell \in \Actions_{\hat\alpha}(s_i)$ in \cref{eq:transition_probability} for brevity.

Two factors prevent us from computing \cref{eq:transition_probability}: 1) the nominal successor state $\hat{\state}_{k+1}$ and the term $\delta(\alpha,\state_k,\control_k)$ are nondeterministic, and 2) the distribution of the noise $\noise_k$ is unknown.
We deal with the former in the following paragraph while addressing the stochastic noise in \cref{sec:PAC_intervals}.

\paragraph{Capturing nondeterminism.}
As a key step, we capture the nondeterminism caused by epistemic uncertainty.
First, recall that by construction, we have $\hat{\state}_{k+1} \in \targetSet_\ell$.
Second, we write the set $\Delta_i$ of all possible epistemic errors in \cref{eq:EpistemicUncertainty_def} as
\begin{equation}
    \Delta_i = \{ \delta(\alpha,\state_k,\control_k) \colon
    \alpha \in \Gamma, 
    \state_k \in \region_i,
    \control_k \in \cControlSpace \}.
    \label{eq:epistemic_error_set}
\end{equation}
Based on these observations, we obtain that the successor state $\state_{k+1}$ is an element of a set that we denote by $\boxSample_{i\ell}$:
\begin{equation}
    \state_{k+1} \in \targetSet_\ell + \Delta_i + \noise_k = \boxSample_{i\ell}.
    \label{eq:NextStateRepresentation}
\end{equation}
Crucially, we show in 
\change{
\ifappendix
    \cref{lemma:EpistemicError_characterization} (provided in \cref{sec:proofs} for brevity) 
\else
    Lemma 4 of the extended version of this paper~\cite[Appendix~A]{Badings2022_AAAI_extended}
\fi }%
that we can compute an \emph{overapproximation of $\Delta_i$} based on sets $\region_i$ and $\cControlSpace$, and the model dynamics.
Based on the set $\boxSample_{i\ell}$, we bound the probability in \cref{eq:transition_probability} as follows:
\begin{equation}
    \label{eq:probability_interval}
    \scalebox{0.97}{$
    \amsmathbb{P} \{ \boxSample_{i\ell} \subseteq \region_j \} \leq
    \amsmathbb{P} \{ \state_{k+1} \in \region_j \} \leq \amsmathbb{P} \{ \boxSample_{i\ell} \cap \region_j \neq \emptyset \}.
    $}
\end{equation}
Both inequalities follow directly by definition of \cref{eq:NextStateRepresentation}.
The lower bound holds, since if $\boxSample_{i\ell} \subseteq \region_j$, then $\state_{k+1} \in \region_j$ for any $\state_{k+1} \in \boxSample_{i\ell}$.
\change{The upper bound holds, since by \cref{eq:NextStateRepresentation} we have that $\state_{k+1} \in \boxSample_{i\ell}$, and thus, if $\state_{k+1} \in \region_j$, then the intersection $\boxSample_{i\ell} \cap \region_j$ must be nonempty.}

\section{PAC Probability Intervals via Sampling}
\label{sec:PAC_intervals} 

The interval in \cref{eq:probability_interval} still depends on the noise $\noise_k$, whose density function is unknown.
We show how to compute \gls{PAC} bounds on this interval, by sampling a set of $N \in \N$ samples of the noise, denoted by $\noise_k^{(1)}, \ldots, \noise_k^{(N)}$.
Recall from \cref{assump:iid_process} that these sample are \gls{iid} elements from $(\Omega, \mathcal{F}, \amsmathbb{P})$.
Each sample $\noise_k^{(\iota)}$ yields a set $\boxSample_{i\ell}^{(\iota)}$ (see \cref{fig:backreachset}) that contains the successor state under that value of the noise, i.e.,
\begin{align}
    &\state_{k+1}^{(\iota)} \in \targetSet_\ell + \Delta_i + \noise^{(\iota)}_k = \boxSample_{i\ell}^{(\iota)}.
    \label{eq:box_sample}
\end{align}
\noindent
\change{For reasons of computational performance, we overapproximate each set $\boxSample_{i\ell}^{(\iota)}$ as the smallest hyperrectangle in $\R^n$, by taking the point-wise min. and max. over the vertices of $\boxSample_{i\ell}^{(\iota)}$.}

\subsection*{Lower Bounds from the Scenario Approach}

We interpret the lower bound in \cref{eq:probability_interval} within the \textit{sampling-and-discarding} scenario approach~\cite{DBLP:journals/jota/CampiG11}.
Concretely, let $\discardedSet \subseteq \{1,\ldots,N\}$ be a subset of the noise samples and consider the following convex program:
\begin{equation}
\begin{split}
    \label{eq:scenario_problem}
    \mathfrak{L}_\discardedSet \colon
    \minimize_{\lambda \geq 0} & \enskip \lambda
    \\
    \mathrm{subject~to} & \enskip \boxSample_{i\ell}^{(\iota)} \subseteq \region_j(\lambda)
    \,\, \quad \forall \iota \in \discardedSet,
\end{split}
\end{equation}
where $\region_j(\lambda)$ is a version of $\region_j$ scaled by a factor $\lambda$ around an arbitrary point $x \in \region_j$, such that $\region_j(0) = x$, and $\region_j(\lambda_1) < \region_j(\lambda_2)$ for $\lambda_1 < \lambda_2$; see~\citet[Appendix~A]{Badings2022AAAI} for details.
The optimal solution $\lambda^\star_\discardedSet$ to $\mathfrak{L}_\discardedSet$ results in a region $\region_j(\lambda^\star_\discardedSet)$ such that, for all $\iota \in \discardedSet$, the set $\mathcal{H}_{i\ell}^{(\iota)}$ for noise sample $\noise_k^{(\iota)}$ is contained in $\mathcal{P}_j(\lambda^\star_\discardedSet)$.
We ensure that $\region_j(\lambda^\star_\discardedSet ) \subseteq \region_j$, by choosing $\discardedSet$ as the set of samples being a subset of $\region_j$, i.e.,
\begin{equation}
    R \coloneqq \big\{ \iota \in \{ 1,\ldots,N\} \colon \mathcal{H}_{i \ell}^{(\iota)} \subseteq \mathcal{P}_j \big\},
    \label{eq:set_R_lower_bound}
\end{equation}
We use the results in~\citet[Theorem~5]{romao2022exact} to lower bound the probability that a random sample $\noise_k \in \Omega$ yields $\boxSample_{i\ell} \subseteq \region_j(\lambda^\star_\discardedSet )$.
This leads to the following lower bound on the transition probability:\footnote{One can readily show that the technical requirements stated in~\citet{romao2022exact} are satisfied for the scenario program \cref{eq:scenario_problem}. Details are omitted here for brevity.}
\begin{lemma}
    \label{lemma:prob_lower_bound}
    Fix a region $\region_j$ and confidence probability $\beta \in (0,1)$.
    Given sets $(\boxSample_{i\ell}^{(\iota)})_{\iota=1}^{N}$, compute $\discardedSet$.
    Then, it holds that
    \begin{equation}
        \label{eq:prob_lower_bound}
        \amsmathbb{P}^N \Big\{ 
            \amsmathbb{P} \{ 
                \noise_k \in \Omega \colon \boxSample_{i\ell} \subseteq \region_j
            \} \geq \munderbar{p}
        \Big\} \geq 1 - \beta,
    \end{equation}
    \noindent
    where $\munderbar{p} = 0$ if $\vert R \vert = 0$, and otherwise, $\munderbar{p}$ is the solution to
    \begin{equation}
        \label{eq:prob_lower_bound2}
        \frac{\beta}{N} = \sum\nolimits_{i=0}^{N - \vert \discardedSet \vert} \binom Ni (1-\munderbar{p})^i \munderbar{p}^{N-i}.
    \end{equation}
\end{lemma}
\noindent
More details and the proof of \cref{lemma:prob_lower_bound} are in 
\change{
\ifappendix
    \cref{sec:proofs}.
\else
    the extended version, \citet[Appendix~A]{Badings2022_AAAI_extended}.
\fi }%

\subsection*{Upper Bounds via Hoeffding's Inequality}

The scenario approach might lead to conservative estimates of the upper bound in \cref{eq:probability_interval}; 
\change{
\ifappendix
    see \cref{sec:proofs} for details.
\else
    see~\citet[Appendix~A]{Badings2022_AAAI_extended} for details.
\fi }%
Thus, we instead apply Hoeffding's inequality~\cite{Boucheron2013concentrationInequalities} to infer an upper bound $\bar{p}$ of the probability $\amsmathbb{P} \{ \boxSample_{i\ell} \cap \region_j \neq \emptyset \}$.
\change{Concretely, this probability describes the parameter of a Bernoulli random variable, which has value $1$ if $\boxSample_{i\ell} \cap \region_j \neq \emptyset$ and $0$ otherwise.}
The sample sum $\tilde{\discardedSet}$ of this random variable is given by the number of sets $\boxSample_{i \ell}^{(\iota)}$ that intersect with region $\region_j$, i.e.,
\begin{equation}
        \tilde{\discardedSet} \coloneqq \big\{ \iota \in \{ 1,\ldots,N \} \colon \mathcal{H}_{i \ell}^{(\iota)} \cap \mathcal{P}_j \neq \emptyset \big\}.
        \label{eq:discardedSet}
\end{equation}
Using Hoeffding's inequality, we state the following lemma to bound the upper bound transition probability in \cref{eq:probability_interval}.
\begin{lemma}
    \label{lemma:prob_upper_bound}
    Fix region $\region_j$ and confidence probability $\beta \in (0,1)$.
    Given sets $(\boxSample_{i\ell}^{(\iota)})_{\iota=1}^{N}$, compute $\tilde{\discardedSet}$.
    Then, it holds that
    \begin{equation}
        \label{eq:prob_upper_bound}
        \amsmathbb{P}^N \Big\{ 
            \amsmathbb{P} \{ 
                \noise_k \in \Omega 
                \colon
                \boxSample_{i\ell} \cap \region_j \neq \emptyset
            \} \leq \bar{p}
        \Big\} \geq 1 - \beta,
    \end{equation}
    where the upper bound $\bar{p}$ is computed as
    \begin{equation}
        \label{eq:prob_upper_bound2}
        \bar{p} = \min \left\{ 1, \enskip \frac{\tilde{\discardedSet}}{N} + \sqrt{\frac{1}{2N} \log\Big(\frac{1}{\beta}\Big)} \right\}.
    \end{equation}
\end{lemma}
\noindent
\change{
\ifappendix
    More details and the proof of \cref{lemma:prob_upper_bound} are in \cref{sec:proofs}.
\else
    We provide the proof in~\citet[Appendix~A]{Badings2022_AAAI_extended}.
\fi }%

\subsection*{Probability Intervals with PAC Guarantees}

We apply \cref{lemma:prob_lower_bound,lemma:prob_upper_bound} as follows to compute a \gls{PAC} probability interval for a specific transition of the \gls{iMDP}.

\begin{theorem}[\gls{PAC} probability interval]
    \label{theorem:probability_interval}
    Fix a region $\region_j$ and a confidence probability $\beta \in (0,1)$.
    For the collection $(\boxSample_{i\ell}^{(\iota)})_{\iota=1}^{N}$, compute $\munderbar{p}$ and $\bar{p}$ using \cref{lemma:prob_lower_bound,lemma:prob_upper_bound}.
    Then, the transition probability $P(s_i, a_\ell)(s_j)$ is bounded by
    \begin{equation}
        \amsmathbb{P}^N \Big\{ \munderbar{p} \leq P(s_i, a_\ell)(s_j) \leq \bar{p} \Big\} \geq 1 - 2 \beta.
        \label{eq:TheoremBounds}
    \end{equation}
\end{theorem}

\begin{proof}
\cref{theorem:probability_interval} follows directly by combining \cref{lemma:prob_lower_bound,lemma:prob_upper_bound} via the union bound\footnote{\change{The union bound (Boole's inequality) states that the probability that at least one of a finite set of events happens, is upper bounded by the sum of these events' probabilities~\cite{casella2021statistical}.}} with the probability interval in \cref{eq:probability_interval}, which asserts that these bounds are both correct with a probability of at least $1-2\beta$.
\end{proof}

\paragraph{Counting samples.}
The inputs to \cref{theorem:probability_interval} are the sample counts $R$ (fully contained in $\region_j$) and $\tilde{\discardedSet}$ (at least partially contained in $\region_j$), and the confidence probability $\beta$.
Thus, the problem of computing probability intervals reduces to a \emph{counting problem} on the samples.
\change{
\ifappendix
    In \cref{sec:approximate_sample_counting}, 
\else
    In \citet[Appendix~B]{Badings2022_AAAI_extended},
\fi }%
we describe a procedure that significantly speeds up this counting process in a sound way by merging (and overapproximating) sets $\boxSample^{(\iota)}$ that are very similar.
\begin{figure}[t!]
	\centering
	\usetikzlibrary{calc, arrows, arrows.meta, shapes, positioning, shapes.geometric}

\tikzstyle{nodewide} = [rectangle, rounded corners, minimum width=3.6cm, text width=3.3cm, minimum height=1.1cm, text centered, draw=black, fill=plotblue!15]
\tikzstyle{nodewide_red} = [rectangle, rounded corners, minimum width=3.6cm, text width=3.3cm, minimum height=1.1cm, text centered, draw=black, ultra thick, fill=red!25]

\resizebox{\linewidth}{!}{%
	\begin{tikzpicture}[node distance=6.0cm,->,>=stealth,line width=0.3mm,auto,main node/.style={circle,draw,font=\sffamily\bfseries}]
		
	\newcommand\yshift{-3.8cm}
	
	\node (dyn) [nodewide_red] 
	{\textbf{Dynamical model} \\ 
	\cref{eq:uncertain_LTI_model}, unknown $\alpha$
	};
	
	\node (abstract) [nodewide, left of=dyn, xshift=0.2cm] 
	{\textbf{Abstract iMDP} \\ $\imdp = \iMDP$};
	
	\node (guarantees) [nodewide, below of=abstract, yshift=3.8cm] 
	{\textbf{Guarantee on iMDP} \\ 
	$\reachProbDiscrRobustStar$};
	
	\node (policy) [nodewide, below of=dyn, yshift=3.8cm] 
	{\textbf{Stored tabular policy} \\ $\pi^\star \colon \States \times \N \to \Actions$ 
	 };
	
	
	\node (in1) [left of=abstract, xshift=1.7cm] {};
	\node (in2) [left of=guarantees, xshift=1.7cm] {};
	
	\draw [->] (in1) -- (abstract) node [midway, above, align=center, font=\sffamily\small] {Partition $(\region_i)_{i=1}^L$} node [midway, below, align=center, font=\sffamily\small] {Confidence $\beta$}; %
	
	\draw [->] (in2) -- (guarantees) node [midway, below, align=center, font=\sffamily\small] {Reach-avoid \\ specification};
	
	\draw [->] (dyn) -- (abstract) node [midway, above, align=center, font=\sffamily\small] {Create \\ abstraction};
	
	
	\draw [->] ($(abstract.south) + (-0.2,0)$) -- ($(guarantees.north) + (-0.2,0)$) node [midway, left, align=center, font=\sffamily\small] {Compute robust \\ optimal policy};
	
	\draw [->, dashed] ($(guarantees.north) + (0.2,0)$) -- ($(abstract.south) + (0.2,0)$) node [midway, right, align=center, font=\sffamily\small] {$\underline{Pr}^{\policy^\star} < \threshold$ \ \rured{\faIcon{times}} \\ Improve abstraction};
	
	\draw [->] (guarantees) -- (policy) 
	node [midway, above, align=center, font=\sffamily\small] {$\underline{Pr}^{\policy^\star} \geq \threshold$ \ \darkgreen{\faIcon{check}}}
	node [midway, below, align=center, font=\sffamily\small] {Extract \\ policy $\policy^\star$};
	
	\draw [->] ($(dyn.south) + (-0.2,0)$) -- ($(policy.north) + (-0.2,0)$) node [midway, left, align=center, font=\sffamily\small] {State $\state_k$};
	
	\draw [->] ($(policy.north) + (0.2,0)$) -- ($(dyn.south) + (0.2,0)$) node [midway, right, align=center, font=\sffamily\small] {Compute $\control_k$ \\ (solve \cref{eq:ControlLaw_opt})};
	
	\draw [-,decorate,decoration={brace,amplitude=5pt,mirror,raise=2ex}] (-3.9cm, 0.5cm) -- (-7.7cm, 0.5cm)
	node[midway,yshift=3em]{\textbf{Offline planning}};
	\draw [-,decorate,decoration={brace,amplitude=5pt,mirror,raise=2ex}] (1.9cm, 0.5cm) -- (-1.9cm, 0.5cm)
	node[midway,yshift=3em]{\textbf{Online control}};
		
	\end{tikzpicture}
}
	\caption{
		Our overall approach to solve the formal problem based on the abstraction method outlined in \cref{sec:Abstractions,sec:PAC_intervals}.
	}
	\label{fig:Approach}
\end{figure}
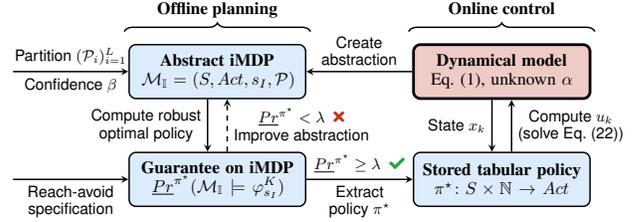

\section{Overall Abstraction Method}
\label{sec:Algorithm}

We provide an algorithm to solve the formal problem based on the proposed abstraction method.
The approach is shown in \cref{fig:Approach} and consists of an \emph{offline planning} phase, in which we create the \gls{iMDP} and compute a robust optimal policy, and an \emph{online control} phase in which we automatically derive a provably-correct controller for the continuous model. 

\paragraph{1) Create abstraction.}
Given a model as in \cref{eq:uncertain_LTI_model}, a partition $(\region_i)_{i=1}^L$ of the state space as defined by \cref{def:Partition}, and a confidence level $\beta \in (0,1)$ as inputs, we create an abstract \gls{iMDP} by applying the techniques presented in \cref{sec:Abstractions,sec:PAC_intervals}.

\paragraph{2) Compute robust optimal policy.}
We compute a robust optimal policy $\policy^\star$ for the \gls{iMDP} using \cref{eq:optimalPolicyRobust}.
Recall that the problem is to find a controller together with a lower bound $\lambda$ on the reach-avoid probability.
If this condition holds, we output the policy and proceed to step 3; otherwise, we attempt to improve the abstraction in one of the following ways.

First, we can refine the partition at the cost of a larger \gls{iMDP}, as shown in \cref{sec:Experiments}.
Second, using more samples $N$ yields an improved \gls{iMDP} through tighter intervals (see, e.g.,~\citet{Badings2022AAAI} for such trade-offs). 
Finally, the uncertainty in $\alpha \in \Gamma$ may be too large, meaning we need to reduce set $\Gamma$ using learning techniques (see the related work).

\paragraph{3) Online control.}
The stored policy is a time-varying map from \gls{iMDP} states to actions.
Recall that an action $a_\ell \in \Actions$ corresponds to applying a control $\control_k$ such that the nominal state $\hat{\state}_{k+1} \in \targetSet_\ell$.
In view of \cref{eq:nominal_model} and \cref{def:Actions}, such a $\control_k \in \cControlSpace$ for the original model \textit{exists by construction} and is obtained as the solution to the following convex optimization program:
\begin{equation}
\begin{split}
    c_\ell(\state_k) = & \argmin_{u \in \cControlSpace} \quad \| A(\hat\alpha) \state_k + B(\hat\alpha) u - \tilde{t}_\ell \|_2
    \label{eq:ControlLaw_opt}
    \\
    & \text{subject to} \quad A(\hat\alpha) \state_k + B(\hat\alpha) u \in \targetSet_\ell,
\end{split}
\end{equation}
\change{with $\tilde{t}_\ell \in \targetSet_\ell$ a representative point, which indicates a point to which we want to steer the nominal state (in practice, we choose $\tilde{t}_\ell$ as the center of $\targetSet_\ell$, but the theory holds for any such point).}
Thus, upon observing the current continuous state $\state_k$, we determine the optimal action $a_\ell$ in the corresponding \gls{iMDP} state and apply the control input $\control_k = c_\ell(\state_k) \in \cControlSpace$.

\subsection*{Correctness of the Abstraction}

We lift the confidence probabilities on individual transitions obtained from \cref{theorem:probability_interval} to a correctness guarantee on the whole \gls{iMDP}.
The following theorem states this key result: 
\begin{theorem}[\change{Correctness of the iMDP}]
    \label{theorem:correctness}
    Generate an \gls{iMDP} abstraction $\imdp$ and compute the robust reach-avoid probability $\reachProbDiscrRobustStar$ under optimal policy $\policy^\star$.
    Under the controller $c$ defined by \cref{eq:ControlLaw_opt} for each $k \leq K$, it holds that
    \begin{equation}
        \label{eq:correctness}
        \amsmathbb{P} \left\{ V(\state_0, \alpha, c) \geq \reachProbDiscrRobustStar \right\} \geq 1 - 2\beta L M.
    \end{equation}
\end{theorem}
\noindent The proof of \cref{theorem:correctness}, 
\change{
\ifappendix
    which we provide in \cref{sec:proofs}, 
\else
    which we provide in \citet[Appendix~A]{Badings2022_AAAI_extended},
\fi }%
uses the union bound with the fact that the \gls{iMDP} has at most $LM$ unique probability intervals.
By tuning $\beta \in (0,1)$, we thus obtain an abstraction that is correct with a user-specified confidence level of $1-\tilde{\beta} = 1-2\beta LM$.

\change{Crucially, we note that \cref{theorem:correctness} is, with a probability of at least $1-2\beta L M$, a solution to the formal problem stated in \cref{sec:Problem} with threshold $\lambda = \reachProbDiscrRobustStar$.}

\paragraph{Sample complexity.}
\change{The required sample size $N$ depends logarithmically on the confidence level, cf. \cref{lemma:prob_lower_bound,lemma:prob_upper_bound}.}
Moreover, the number of unique intervals is often lower than the worst-case of $LM$, so the bound in  \cref{theorem:correctness} can be conservative.
In particular, the number of intervals only depends on the state and action definitions of the \gls{iMDP}, and is thus observable before we apply \cref{theorem:probability_interval}.
To compute less conservative probability intervals, we replace $LM$ in \cref{eq:correctness} with the observed number of unique intervals. 

\subsection*{Uncertain Additive Disturbance}

\change{We extend the generality of our models with an additive parameter representing an external disturbance $\disturbance_k \in \mathcal{Q}$} that, in the spirit of this paper, belongs to a convex set $\mathcal{Q} \subset \R^n$~\cite{blanchini2008set}. 
The resulting model is
\begin{equation}
    \state_{k+1} = A(\alpha) \state_k + B(\alpha) \control_k + \disturbance_k + \noise_k.
    \label{eq:uncertain_LTI_model_disturbance}
\end{equation}
\change{This additional parameter $\disturbance_k$ models uncertain disturbances \emph{that are not stochastic} (and can thus not be captured by $\noise_k$), and that are independent of the state $\state_k$ and control $\control_k$ (and can thus not be captured in $A(\alpha)$ or $B(\alpha)$).}
To account for parameter $q_k$ (which creates another source of epistemic uncertainty), we expand \cref{eq:EpistemicUncertainty_def} and 
\change{
\ifappendix
    \cref{lemma:EpistemicError_characterization}
\else
    Lemma 4 of \citet[Appendix~A]{Badings2022_AAAI_extended},
\fi }%
\emph{to be robust against any $\disturbance_k \in \mathcal{Q}$}.
While this extension increases the size of the sets $\boxSample_{i\ell}^{(\iota)}, \iota = 1,\ldots,N$, the procedure outlined in \cref{sec:PAC_intervals} to compute probability intervals remains the same.

\paragraph{Generality of the model.}
The parameter $q_k$ expands the applicability of our approach significantly.
Consider, e.g., a building temperature control problem, where only the temperatures of adjacent rooms affect each other.
We can decompose the model dynamics into the individual rooms by capturing any possible influence between rooms into $\disturbance_k\in\mathcal{Q}$, as is common in assume-guarantee reasoning~\cite{DBLP:conf/cav/BobaruPG08}.
We apply this extension to a large building temperature control problem in \cref{sec:Experiments}.
\section{Numerical Experiments}
\label{sec:Experiments}

We perform experiments to answer the question: \textquotedblleft\emph{Can our method synthesize controllers that are robust against epistemic uncertainty in parameters?}\textquotedblright\,
In this section, we focus on problems from motion planning and temperature control, and we discuss an additional experiment on a variant of the automated anesthesia delivery benchmark 
from~\citet{DBLP:conf/adhs/AbateBCHHLOSSVV18} 
\change{
\ifappendix
    in \cref{sec:details_experiments}.
\else
    in the extended version of this paper~\cite[Appendix~C]{Badings2022_AAAI_extended}.
\fi }%
All experiments ran single-threaded on a computer with $32$ $3.7$GHz cores and $64$GB RAM.
A Python implementation of our approach is available at \color{Sepia}\url{https://github.com/LAVA-LAB/DynAbs}\color{black}, using the probabilistic model checker PRISM~\cite{DBLP:conf/cav/KwiatkowskaNP11} to compute optimal \gls{iMDP} policies.

\begin{figure}[t!]
    \centering
    \includegraphics[width = \columnwidth]{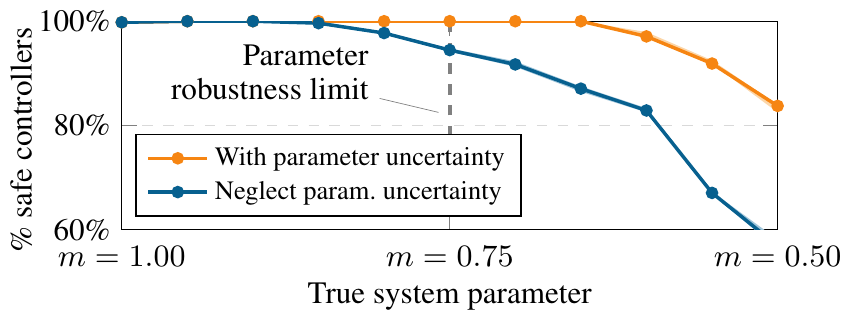}
    \caption{Percentage of initial states with safe performance guarantees (i.e., the simulated reach-avoid probability is above the optimum of \cref{eq:optimalPolicyRobust} on the \gls{iMDP}).
    Our approach that accounts for epistemic uncertainty is \emph{100\% safe up to the parameter robustness limit}; neglecting uncertainty is not.}    
    \label{fig:harmonic_oscillator_plot}
\end{figure}

\begin{figure}[t!]
    \centering
    \includegraphics[width = .9\columnwidth]{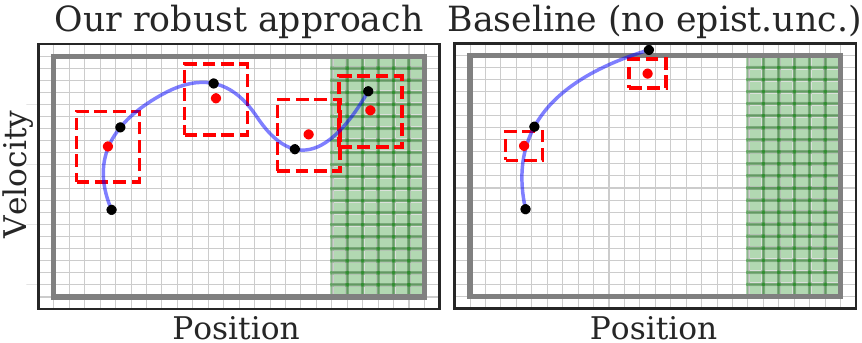}
    \caption{With our approach, the system safely reaches the goal (in green), while the baseline neglecting epistemic uncertainty leaves the safe set (gray box), as it underestimates the successor state sets $\boxSample_{i\ell}$ defined in \cref{eq:NextStateRepresentation} (red boxes).
    }    
    \label{fig:harmonic_oscillator_traces}
\end{figure}

\subsection*{Longitudinal Drone Dynamics}

We revisit \cref{ex:Model} of a drone with an uncertain mass $m \in [0.75, 1.25]$.
We fix the nominal value of the mass as $\hat{m} = 1$.
To purely show the effect of epistemic uncertainty, we set the covariance of the aleatoric uncertainty in $\noise_k$ (being a Gaussian distribution) to almost zero.
The specification is to reach a position of $p_k \geq 8$ before time $K = 12$, while avoiding speeds of $\vert v_k \vert \geq 10$.
\change{Thus, the safe set is $\SafeSet = [-\infty, \infty] \times [-10, 10]$, of which we create a partition covering $\partitionSpace = [-10, 14] \times [-10, 10]$ and into $24 \times 20$ regions.}
We use $20$K samples of the noise to estimate probability intervals.
We compare against a baseline that builds an \gls{iMDP} for the nominal model only, thus neglecting parameter uncertainty.

\paragraph{Neglecting epistemic uncertainty is unsafe.}
\change{
We solve the formal problem stated in \cref{sec:Problem} for every initial state $\state_0 \in \partitionSpace$, resulting in a threshold $\lambda$ for each of those states.}
The run time for solving this benchmark is around $\SI{3}{\second}$.
\change{For each $\state_0$, we say that the controller $c$ at a parameter value $\alpha \in \Gamma$ is \emph{unsafe},} if the reach-avoid probability $V(\state_0, \alpha, c)$ (estimated using Monte Carlo simulations) is below $\lambda = \reachProbDiscrRobustStar$ on the \gls{iMDP}, as per \cref{eq:optimalPolicyRobust}.
In \cref{fig:harmonic_oscillator_plot}, we show the deviation of the actual mass $m$ from its nominal value, versus the average percentage of states with a safe controller (over 10 repetitions).
The \emph{parameter robustness limit} represents the extreme values of the parameter against which our approach is guaranteed to be robust ($m=0.75$ and $1.25$ in this case).

Our approach yields \emph{100\% safe controllers} for deviations well over the robustness limit, while the baseline yields \emph{6\% unsafe controllers} at this limit.
We show simulated trajectories under an actual mass $m=0.75$ in \cref{fig:harmonic_oscillator_traces}.
These trajectories confirm that our approach safely reaches the goal region while the baseline does not, as it neglects epistemic uncertainty, which is equivalent to assuming $\Delta_i = 0$ in \cref{eq:NextStateRepresentation}.

\change{
\paragraph{Multiple uncertain parameters.}
To show that our contributions hold independently of the uncertain parameter, we consider a case in which, in addition to the uncertain mass, we have an uncertain friction coefficient.
The results of this experiment, presented
\change{
\ifappendix
    in \cref{sec:details_experiments},
\else
    in~\citet[Appendix~C]{Badings2022_AAAI_extended},
\fi }%
show that we obtain controllers with similar correctness guarantees, irrespective of the number of uncertain~parameters.}

\subsection*{Building Temperature Control}

We consider a temperature control problem for a 5-room building with Gaussian process noise, each with a dedicated radiator that has an uncertain power output of $\pm 10\%$ around its nominal value 
\change{
\ifappendix
    (see \cref{sec:details_experiments} for modeling details).
\else
    (see~\citealt[Appendix~C]{Badings2022_AAAI_extended} for details).
\fi }%
The $10$D state of this model captures the temperatures of 5 rooms and 5 radiators.
The goal is to maintain a temperature within $21 \pm \SI{2.5}{\celsius}$ for $15$ steps of $\SI{20}{\minute}$.  

\change{
\paragraph{Interactions between rooms as nondeterminism.}
Since a direct partitioning of the $10$D state space is infeasible, we use the procedure from \cref{sec:Algorithm} to capture \emph{any possible thermodynamic interaction} between rooms in the uncertain parameter $\disturbance_k \in \mathcal{Q}$.
In particular, we show in 
\ifappendix
    \cref{sec:details_experiments}
\else
    \citet[Appendix~C]{Badings2022_AAAI_extended}
\fi %
that the set $\mathcal{Q}_i$ for room $i \in \{1,\ldots,5\}$ is characterized by the maximal temperature difference between room $i$ and all adjacent rooms.
For the partition used in this particular reach-avoid problem, this maximal temperature difference is $23.5 - 18.5 = \SI{5}{\celsius}$.
Following this procedure, we can easily derive a set-bounded representation of $\mathcal{Q}$, allowing us to decouple the dynamics into the individual rooms.
}

\begin{figure}[t!]
    \centering
    \includegraphics[width = \columnwidth]{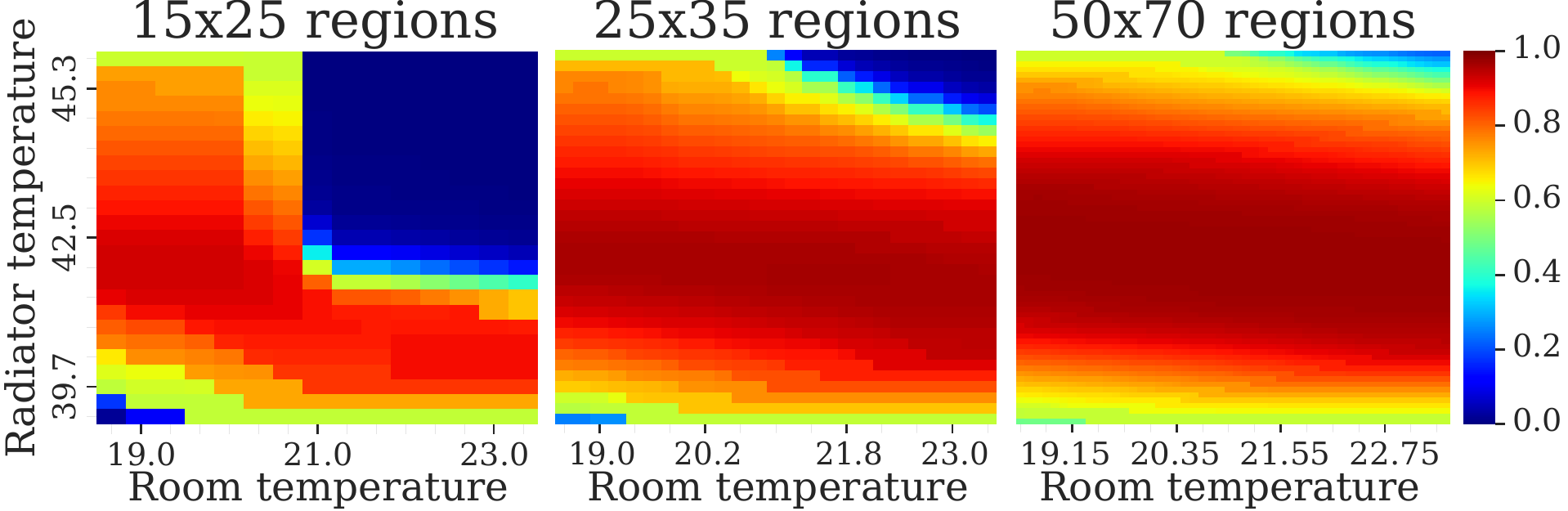}
    \caption{
    Specification satisfaction probabilities on the \gls{iMDP} for a single room of the temperature control problem.}
    \label{fig:temperature_control_heatmaps}
\end{figure}

\paragraph{Refining partitions improves results.}
\change{
We apply our method with an increasingly more fine-grained state-space partition.
In \cref{fig:temperature_control_heatmaps}, we present, for three different partitions, the thresholds $\lambda$ of satisfying the specification under the robust optimal \gls{iMDP} policy as per \cref{eq:optimalPolicyRobust}, from any initial state $\state_0 \in \partitionSpace$.}
These results confirm the idea from \cref{sec:Algorithm} that partition refinement can lead to controllers with better performance guarantees.
A more fine-grained partition leads to more actions enabled in the abstraction, which in turn improves the robust lower bound on the reach-avoid probability.

\paragraph{Scalability.}
We report run times and model sizes 
\change{
\ifappendix
    in \cref{tab:temperature_control} in \cref{sec:details_experiments}.
\else
    in~\citet[Appendix~C,~Table~1]{Badings2022_AAAI_extended}.
\fi }%
The run times vary between $\SI{5.6}{\second}$ and $\SI{8}{\minute}$ for the smallest ($15 \times 25$) and largest ($70 \times 100$) partitions, respectively.
Without the decoupling procedure, even a 2-room version on the smallest partition \emph{leads to a memory error}.
By contrast, our approach with decoupling procedure has \emph{linear complexity in the number of rooms}.
We observe that accounting for epistemic uncertainty yields \glspl{iMDP} with more transitions and slightly higher run times (for the largest partition: $82$ instead of $52$ million transitions and $8$ instead of $\SI{5}{\minute}$).
This is due to larger successor state sets $\boxSample_{i\ell}^{(\iota)}$ in \cref{eq:NextStateRepresentation} caused by the epistemic error $\Delta_i$.
\section{Conclusions and Future Work}
\label{sec:Conclusion}

We have presented a novel abstraction-based controller synthesis method for dynamical models with aleatoric and epistemic uncertainty.
The method captures those different types of uncertainties in order to ensure certifiably safe controllers.
Our experiments show that we can synthesize controllers that are robust against uncertainty and, in particular, against deviations in the model parameters.

\change{
\paragraph{Generality of our approach.}
We stress that the models in \cref{eq:uncertain_LTI_model,eq:uncertain_LTI_model_disturbance} can capture many common sources of uncertainty.
As we have shown, our approach simultaneously deals with epistemic uncertainty over one or multiple parameters, as well as aleatoric uncertainty due to stochastic noise of an unknown distribution.
Moreover, the additive parameter $q_k$ introduced in \cref{sec:Algorithm} enables us to generate abstractions that faithfully capture any error term represented by a bounded set (as we have done for the temperature control problem).
%
For example, we plan to apply our method to \emph{nonlinear systems}, such as non-holonomic robots~\cite{DBLP:books/daglib/Thrun2005}.
Concretely, we may apply our abstraction method on a linearized version of the system while treating linearization errors as nondeterministic disturbances $\disturbance_k \in \mathcal{Q}$ in \cref{eq:uncertain_LTI_model_disturbance}.
The main challenge is then to obtain this set-bounded representation $\mathcal{Q}$ of the linearization error.

\paragraph{Scalability.}
Enforcing robustness against epistemic uncertainty hampers the scalability of our approach, especially compared to similar non-robust abstraction methods, such as~\citealt{Badings2022AAAI}.
To reduce the computational complexity, we restrict partitions to be rectangular, but the theory is valid for any convex partition.
A finer partition yields larger \glspl{iMDP} but also improves the guarantees on controllers.

\paragraph{Safe learning.}
Finally, we wish to integrate the abstractions in a \emph{safe learning framework}~\cite{Brunke2021safelearning} by, as discussed in the related work in \cref{sec:Introduction}, applying our approach to guarantee safe interactions with the system.
}

\section*{Acknowledgments}
This research has been partially funded by NWO grant NWA.1160.18.238 (PrimaVera) and by EPSRC IAA Award EP/X525777/1.

\bibliography{references}
\clearpage

\ifappendix
    \appendix
    \section{Proofs}
\label{sec:proofs}

\subsection*{Computing Backward Reachable Sets}
The following lemma shows that computing backward reachable sets, as required for \cref{def:Actions} is computationally tractable.
Because $\targetSet_\ell$ and $\cControlSpace$ are convex polytopes, the inverse of the dynamics of the nominal model in \cref{eq:nominal_model} is also a convex polytope, which is exactly characterized by the vertices of $\cControlSpace$ and $\targetSet_\ell$.
We formalize this claim in \cref{lemma:BackReachSet_Representation}.

\begin{lemma}[Representation of the backward reachable set] 
    \label{lemma:BackReachSet_Representation}
    Let the control space $\cControlSpace = \mathrm{conv} (u^1, \ldots, u^q), \, q \in \N$, and target set $\targetSet_\ell = \mathrm{conv}(t^1,\ldots,t^d), \, d \in \N$, of action $a_\ell$ be given in their vertex representations.
    Under \cref{assump:nominal_nonsingular}, we have that
    \begin{equation}
        \BackReach(\targetSet_\ell) = \mathrm{conv}(\bar{x}_{ij} \colon i=1, \ldots, d,~j = 1, \ldots, q),
        \label{eq:BackReachSet_representation}
    \end{equation}
    where $\bar{x}_{ij}$ is the unique solution of the linear system 
    \begin{equation}
        A(\hat \alpha) \bar{x}_{ij} + B(\hat \alpha) u^j = t^i.
        \label{eq:BackReachSet_lin_sys}
    \end{equation}
\end{lemma}

\begin{proof}
We first proof that \cref{eq:BackReachSet_representation} holds with inclusion, i.e.,
\begin{equation}
    \mathrm{conv}(\bar{x}_{ij} \colon i=1, \ldots, d,~j = 1, \ldots, q) \subseteq \BackReach(\targetSet_\ell).
    \label{eq:proof_lemma_BackReachSet_incl}
\end{equation}
Let $z$ be any element belonging to the right-hand side of \cref{eq:BackReachSet_representation}, i.e. $z \in \mathrm{conv}(\bar{x}_{ij} \colon i=1, \ldots, d,~j = 1, \ldots, q)$.
Then, there exists $\gamma_{ij}$, $i  = 1, \ldots, d, j = 1, \ldots, q$, such that
\begin{equation*}
    \gamma_{ij} \geq 0, \quad \sum_{i,j = 1}^{d,q} \gamma_{ij} = 1, \quad
    z = \sum_{i,j =1}^{d,q} \gamma_{ij} \bar{x}_{ij}.
\end{equation*}
For each vertex $u^j$ of $\cControlSpace$, $j = 1\ldots, q$, let $\xi_j = \sum_{i = 1}^d \gamma_{ij}$ and write the control input $u$ corresponding to point $z$:
\begin{equation*}
    u = \sum_{j = 1}^q \xi_j u^j, \quad u \in\cControlSpace,
    \label{eq:proof_lemma_BackReachSet}
\end{equation*}
which is admissible by construction. 
Now, note that the mapping of the pair $(z,u)$ under the dynamics satisfies
\begin{align}
    A(\hat \alpha) z + B(\hat \alpha) u &= \sum_{i,j = 1}^{d,q} \gamma_{ij} A(\hat \alpha) \bar{x}_{ij} + \sum_{j = 1}^q \xi_j B(\hat \alpha) u^j \nonumber 
    \\ 
    &= \sum_{i,j = 1}^{d,q} \gamma_{ij} \left( A(\hat \alpha) \bar{x}_{ij} +  B(\hat \alpha) u^j \right)  \nonumber 
    \\
    &= \sum_{i = 1}^{d} \bar\gamma_{i} t^i \in \targetSet_\ell,  \nonumber
\end{align}
where the first equality follows from the definition of $z$ and $u$, the second by the definition of $\xi_j = \sum_{i = 1}^d \gamma_{ij}$, and the third by letting $\bar\gamma_i = \sum_{j =1}^q \gamma_{ij}$ and noting that $\sum_{i = 1}^d \bar \gamma_i = 1$ and $\bar \gamma_i \geq 0$ for all $i = 1, \ldots, d$. 
In other words, the mapping of the pair $(z,u)$ belongs to the target set $\targetSet_\ell$, which implies that \cref{eq:proof_lemma_BackReachSet_incl} holds by construction.
This concludes the first part of the lemma.

To show the opposite direction in \cref{eq:BackReachSet_representation}, let $z$ be any element in $\BackReach(\targetSet_\ell)$. 
By definition of the backward reachable set (see \cref{sec:Abstractions}), this means that there exist $\xi_j \geq 0$ and $\gamma_i \geq 0$, $i = 1, \ldots, d, j = 1, \ldots, q$, with $\sum_{j = 1}^q \xi_j = 1$ and $\sum_{i = 1}^d \gamma_i = 1$, such that 
\begin{equation}
A(\hat \alpha) z + B(\hat \alpha) \left( \sum_{j = 1}^q \xi_j u^j \right) = \sum_{i = 1}^d \gamma_i t^i.
\label{eq:second-inclusion-proof}
\end{equation}
In other words, if $z \in \BackReach(\targetSet_\ell)$ then there exists an input $u \in \cControlSpace$ such that $A(\hat \alpha) z + B(\hat \alpha) u \in \targetSet_\ell$.
Substituting \eqref{eq:BackReachSet_lin_sys} into \eqref{eq:second-inclusion-proof} we obtain
\begin{equation}
\begin{split}
    & A(\hat \alpha) z + B(\hat \alpha) \sum_{j = 1}^q \xi_j u^j = \sum_{i = 1}^d
    \gamma_i \left( A(\hat \alpha) \bar{x}_{ik} + B(\hat \alpha) u^k \right)
    \\
    & A(\hat \alpha) z = \sum_{i = 1}^d
    \gamma_i \left( A(\hat \alpha) \bar{x}_{ik} + B(\hat \alpha) u^k \right) - B(\hat \alpha) \sum_{j = 1}^q \xi_j u^j,
\label{eq:second-inclusion-proof-2}
\end{split}
\end{equation}
for all $k = 1, \ldots, q$. Multiplying both sides of \eqref{eq:second-inclusion-proof-2} by $A(\hat \alpha)^{-1}$, which is allowed due to Assumption \ref{assump:nominal_nonsingular}, yields
\begin{equation}
\begin{split}
    z = \sum_{i = 1}^d
    \gamma_i \big( \bar{x}_{ik} & + A(\hat \alpha)^{-1} B(\hat \alpha) u^k \big) 
    \\
    & - A(\hat \alpha)^{-1} B(\hat \alpha) \sum_{j = 1}^q \xi_j u^j.
    \label{eq:second-inclusion-proof-2b}
\end{split}
\end{equation}
Since \cref{eq:second-inclusion-proof-2b} holds for all $k=1,\ldots,q$, we can multiply both sides by $\xi_k$ for each $k = 1, \ldots, q$, and sum up the resulting expression. Due to the fact that $\sum_{k = 1}^q \xi_k = 1$, we obtain
\begin{equation}
\begin{split}
    z = \sum_{i,k = 1}^{d,q} \bar{\gamma}_{ik} \bar{x}_{ik} &+ A(\hat \alpha)^{-1} B(\hat \alpha) \sum_{k = 1}^q \xi_k u^k 
    \\ 
    &- A(\hat \alpha)^{-1} B(\hat \alpha) \sum_{j = 1}^q \xi_j u^j,
    \label{eq:second-inclusion-proof-3}
\end{split}
\end{equation}
where $\bar{\gamma}_{ik} = \xi_k \gamma_i$, which is larger than or equal to zero for all $i = 1, \ldots, d, k = 1, \ldots, q$.
Since the last two terms on the right-hand side of \eqref{eq:second-inclusion-proof-3} cancel out and $\sum_{i,k  = 1}^{d,q} \bar{\gamma}_{ik} = 1$, we conclude that $z \in \mathrm{conv}(\bar{x}_{ij}: i = 1, \ldots, d, j = 1,\ldots, q)$, thus proving the opposite inclusion and concluding the proof of the lemma.
\end{proof}

\subsection*{Bounding the Epistemic Error}
The following lemma shows that the set $\Delta_i$, defined in \cref{eq:epistemic_error_set} is a subset of a convex polytope, which is characterized by the region $\region_i$, the feasible control space $\cControlSpace$, and the model dynamics.
Importantly, note that the probability interval in \cref{eq:probability_interval} also holds for any overapproximation of $\Delta_i$ obtained using \cref{lemma:EpistemicError_characterization}.

\begin{lemma}
\label{lemma:EpistemicError_characterization}
    Given the vertex representations of sets $\region_i = \mathrm{conv}(v^1, \ldots, v^p)$ and $\cControlSpace = \mathrm{conv}(u^1,\ldots,u^q)$ for $p,q \in \N$, define $\Delta_i$ as in \cref{eq:epistemic_error_set}. 
    Then, we have that
    \begin{equation}
    \begin{split}
        \Delta_i 
        \subseteq \mathrm{conv}\Big(& \big(A_\iota - A(\hat\alpha) \big) v^j + \big(B_\iota - B(\hat\alpha) \big) u^{\ell}
        \\
        & \colon \iota,j,\ell= 1, \ldots, \{r,p,q\} \Big),  
        \label{eq:EpistemicError_characterization}
    \end{split}
    \end{equation}
    where $A_1,\ldots,A_r$ and $B_1,\ldots,B_r$ are defined in \cref{eq:dynamics_matrix_vertices}.
\end{lemma}

\begin{proof}

First, let us fix any $\alpha \in \Gamma$, and observe that the set $\Delta_i$ defined in \cref{eq:epistemic_error_set} evaluated at $\alpha$ is written as
\begin{equation}
\begin{split}
    \Delta_i (\alpha)
    &= \left\{ \delta(\alpha,\state_k,\control_k) \colon
    \state_k \in \region_i,
    \control_k \in \cControlSpace \right\}
    \\
    &= \Big\{ (A(\alpha) - A(\hat{\alpha}))\state_k + (B(\alpha) - B(\hat{\alpha}))\control_k
    \\ &\quad\quad
    \colon
    \state_k \in \region_i,
    \control_k \in \cControlSpace \Big\}.
    \label{eq:epistemic_error_fixed_alpha}
\end{split}
\end{equation}
We observe that the sets $\{ (A(\alpha) - A(\hat{\alpha}))\state_k \colon x_k \in \region_i \}$ and $\{ (B(\alpha) - B(\hat{\alpha}))\control_k \colon u_k \in \cControlSpace \}$ are both convex polytopes characterized by the vertices of $\region_i$ and $\cControlSpace$, respectively.
Thus, we rewrite \cref{eq:epistemic_error_fixed_alpha} as
\begin{equation*}
\begin{split}
    \Delta_i (\alpha)
    = \mathrm{conv}\Big( & (A(\alpha) - A(\hat{\alpha})) v^j + (B(\alpha) - B(\hat{\alpha})) u^{\ell}
    \\
    & \enskip\colon j,\ell= 1, \ldots, \{p,q\} \Big).
    \label{eq:epistemic_error_fixed_alpha2}
\end{split}
\end{equation*}
Note that the full set $\Delta_i$ is the union of $\Delta_i (\alpha)$ over all $\alpha \in \Gamma$:
\begin{equation}
\begin{split}
    \Delta_i &= \bigcup_{\alpha \in \Gamma} \Delta_i (\alpha)
    \\
    &\subseteq \mathrm{conv}\Big( (A(\alpha) - A(\hat{\alpha})) v^j + (B(\alpha) - B(\hat{\alpha})) u^{\ell}
    \\
    & \qquad\qquad\enskip \colon j,\ell= 1, \ldots, \{p,q\}, \, \alpha \in \Gamma \Big).
    \label{eq:epistemic_error_fixed_alpha3}
\end{split}
\end{equation}
Crucially, observe that for any fixed pair of vertices $\bar{v} \coloneqq v^j$ and $\bar{u} \coloneqq u^\ell$, $j,\ell = 1,\ldots,\{p,q\}$, we can write the convex hull in \cref{eq:epistemic_error_fixed_alpha3} in terms of only the matrices $A_\iota, B_\iota$ for $\iota=1,\ldots,r$ of which $A(\alpha)$ and $B(\alpha$ are a convex combination (as defined in \cref{eq:dynamics_matrix_vertices}):
\begin{align}
    & \mathrm{conv}\big( (A(\alpha) - A(\hat{\alpha})) \bar{v} + (B(\alpha) - B(\hat{\alpha})) \bar{u} \colon \alpha\in\Gamma \big)
    \nonumber
    \\
    & \enskip= \mathrm{conv}\Big( (A(\alpha) - A(\hat{\alpha})) \bar{v} + (B(\alpha) - B(\hat{\alpha})) \bar{u}
    \nonumber
    \\
    & \qquad\qquad\quad \colon \alpha = e_1, \ldots, e_r \Big)
    \label{eq:epistemic_error_fixed_others}
    \\
    & \enskip= \mathrm{conv}\big( (A_\iota - A(\hat{\alpha})) \bar{v} + (B_\iota - B(\hat{\alpha})) \bar{u} \colon \iota = 1,\ldots,r \big),
    \nonumber
\end{align}
where $e_\iota \in \R^r$ is the vector with all components equal to $0$, except the $\iota^\text{th}$, which is $1$.
The last equality in \cref{eq:epistemic_error_fixed_others} holds, since $A(e_\iota) = A_\iota$ and $B(e_\iota) = B_\iota$ for any $\iota=1\ldots,r$.
In other words, considering the values $\alpha \in \Gamma \setminus \{ e_1,\ldots,e_r \}$ in \cref{eq:epistemic_error_fixed_alpha3} is redundant since these values can be expressed as a convex combination of $\alpha \in \{e_1,\ldots,e_r\}$, as in \cref{eq:epistemic_error_fixed_others}.
As a result, we simplify \cref{eq:epistemic_error_fixed_alpha3} as
\begin{equation*}
    \begin{split}
        \Delta_i
        \subseteq \mathrm{conv}\Big(& \big(A_\iota - A(\hat\alpha) \big) v^j + \big(B_\iota - B(\hat\alpha) \big) u^{\ell}
        \\
        & \colon \iota,j,\ell= 1, \ldots, \{r,p,q\} \Big),  
    \end{split}
    \end{equation*}
which equals \cref{eq:EpistemicError_characterization}.
This concludes the proof of \cref{lemma:EpistemicError_characterization}.
\end{proof}

\paragraph{\cref{eq:EpistemicError_characterization} does not hold with equality.}
It may be tempting to conclude that \cref{eq:EpistemicError_characterization} holds with equality.
However, we show with a simple example that this is not the case.
Specifically, we apply \cref{lemma:EpistemicError_characterization} to a model with the matrices
\begin{equation*}
\begin{split}
    A_1 &= \begin{bmatrix}
        0.9 & 1 \\ 0 & 0.9
    \end{bmatrix},
    \,
    A_2 = \begin{bmatrix}
        1.1 & 1 \\ 0 & 1.1
    \end{bmatrix},
    \,
    A(\hat\alpha) = \begin{bmatrix}
        1 & 1 \\ 0 & 1
    \end{bmatrix},
    \\
    B_1 &= B_2 = B(\hat\alpha) = \begin{bmatrix}
        1 \\ 1
    \end{bmatrix},
    \,\,
    \region_i = [0, 1]^2,
    \,\,
    \cControlSpace = [-5, 5].
\end{split}
\end{equation*}
The resulting right-hand side of \cref{eq:EpistemicError_characterization} is shown by the dashed hull in \cref{fig:contradiction_epistError}.
Moreover, to approximate $\Delta_i$, we compute $\delta(\alpha,\state_k,\control_k)$ for many linearly spaced points $\alpha \in \Gamma$, $\state_k \in \region_i$, and $\control_k \in \cControlSpace$, which are shown by the blue dots in \cref{eq:EpistemicError_characterization}.
While the convex hull is a sound overapproximation of the set $\Delta_i$, \emph{the opposite is clearly not the case} (there are points in the convex hull that are not included in $\Delta_i$).
This result empirically shows that \cref{eq:EpistemicError_characterization} of \cref{lemma:EpistemicError_characterization} does not hold with equality.

\begin{figure}[t!]
    \centering
    \includegraphics[width = .75\columnwidth]{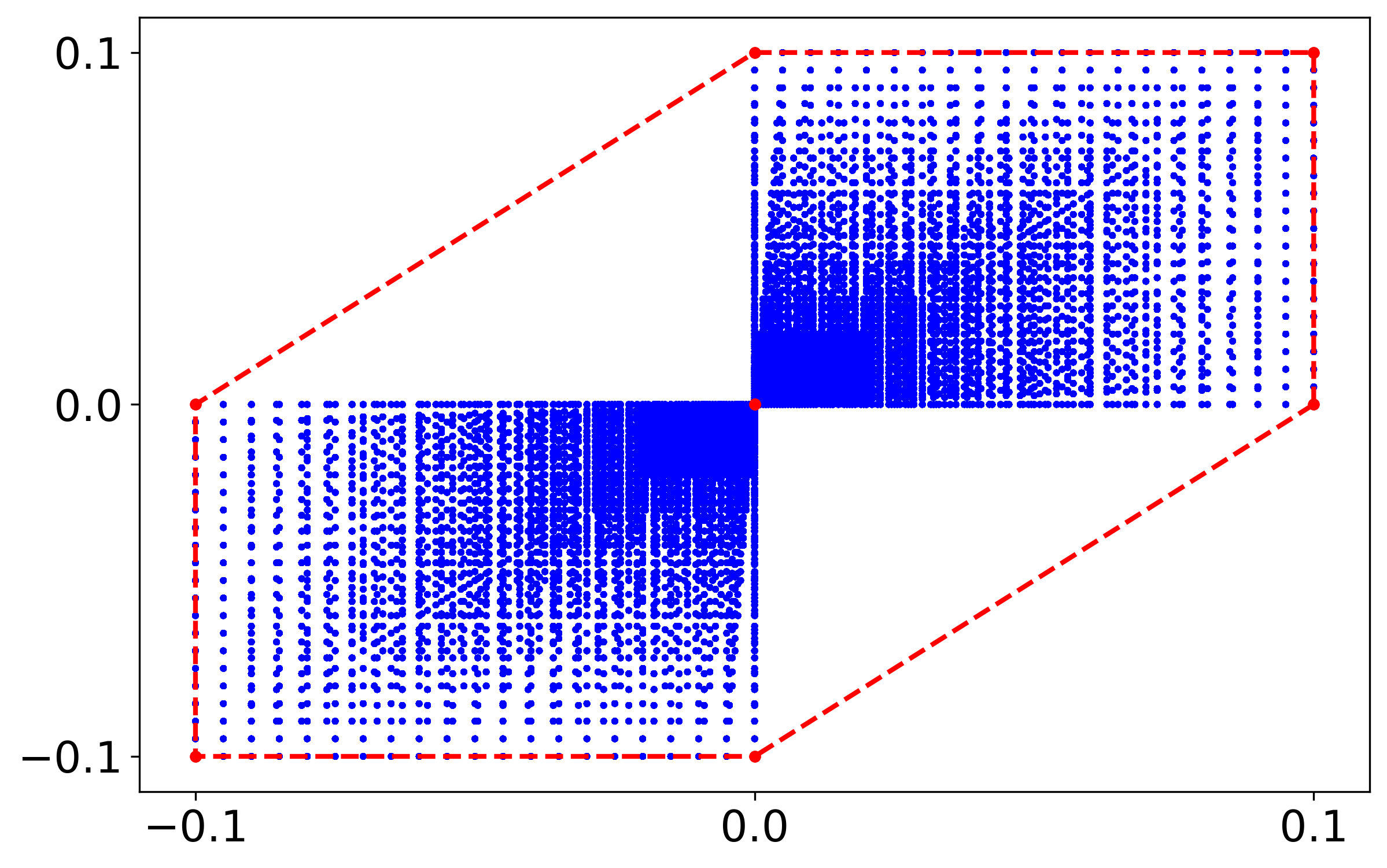}
    \caption{Overapproximation of set $\Delta_i$ using \cref{lemma:EpistemicError_characterization}, shown as the red dashed convex hull, and many points in the set $\Delta_i$ for different $\state_k \in \region_i$ and $\control_k \in \cControlSpace$, shown in blue.}
    \label{fig:contradiction_epistError}
\end{figure}

\subsection*{Proof of \cref{lemma:prob_lower_bound} (Lower Bound Probability)}
Our goal is to show that if we remove $\discardedSet$ samples as defined in \cref{eq:set_R_lower_bound}, the optimal solution to the scenario optimization problem $\mathfrak{L}_{\discardedSet}$ satisfies the claim in \cref{lemma:prob_lower_bound}.
To this end, we use the key result from~\citet[Theorem~5]{DBLP:conf/cdc/RomaoMP20}, which requires three key assumptions on the scenario problem $\mathfrak{L}_{\discardedSet}$:
\begin{enumerate}
    \item Problem $\mathfrak{L}_{\discardedSet}$ belongs to the class of so-called \emph{fully-supported problems} (see~\citet{DBLP:journals/siamjo/CampiG08} for a definition);
    \item The solution to problem $\mathfrak{L}_{\discardedSet}$ is \emph{unique};
    \item All samples \emph{not contained} in $\discardedSet$ violate the optimal solution with probability one.
\end{enumerate}
Requirement (1) is satisfied because problem $\mathfrak{L}_{\discardedSet}$ has a scalar decision variable, (2) is satisfied due to \cref{assump:iid_process}, and (3) is satisfied by definition of $\discardedSet$, since any sample not contained in $\discardedSet$ is violated by the solution to $\mathfrak{L}_\discardedSet$ with probability one.
We invoke the key result by~\citet[Theorem~5]{DBLP:conf/cdc/RomaoMP20} that, under these requirements, the optimal solution $\lambda^\star_\discardedSet$ satisfies the following: 
\begin{equation}
\begin{split}
    \amsmathbb{P}^N \Big\{ & \amsmathbb{P} 
        \{
        \noise_k \in \Omega \colon
        \boxSample_{i \ell} \not\subseteq \region_j(\lambda^\star_\discardedSet) \} \leq \varepsilon \Big\}
    \\
    & \qquad = 1 - \sum_{i=0}^{N - \vert \discardedSet \vert} \binom Ni \varepsilon^i (1-\varepsilon)^{N-i},
    \label{eq:scenario_proof_1}
\end{split}
\end{equation}
where $\varepsilon$ is an upper bound on the so-called \emph{violation probability}, which is the probability that a sample $\boxSample_{i\ell}$ under a random noise value $\noise_k \in \Omega$ is not fully contained in the feasible set $\region_j(\lambda^\star_\discardedSet)$ to problem $\mathfrak{L}_\discardedSet$.
Intuitively, \cref{eq:scenario_proof_1} states that the violation probability is bounded by $\varepsilon$, and this result holds with the confidence at the right-hand side.
Note that we can rewrite this violation probability as the \emph{satisfaction probability} as:
\begin{equation*}
    \amsmathbb{P}\{ \boxSample_{i \ell} \not\subseteq \region_j(\lambda^\star_\discardedSet) \}
    =
    1 - \amsmathbb{P}\{ \boxSample_{i \ell} \subseteq \region_j(\lambda^\star_\discardedSet) \}.
\end{equation*}
Moreover, by defining $\munderbar{p} = 1-\varepsilon$, we rewrite \cref{eq:scenario_proof_1} as
\begin{equation}
\begin{split}
    \amsmathbb{P}^N \Big\{ & 1 - \amsmathbb{P} 
        \{
        \noise_k \in \Omega \colon
        \boxSample_{i \ell} \subseteq \region_j(\lambda^\star_\discardedSet) \} \leq 1-\munderbar{p} \Big\}
    \\
    & = \amsmathbb{P}^N \Big\{ \amsmathbb{P} 
        \{
        \noise_k \in \Omega \colon
        \boxSample_{i \ell} \subseteq \region_j(\lambda^\star_\discardedSet) \} \geq \munderbar{p} \Big\}
    \\
    & = 1 - \sum_{i=0}^{N - \vert \discardedSet \vert} \binom Ni (1-\munderbar{p})^i \munderbar{p}^{N-i} = 1 - \frac{\beta}{N}.
    \label{eq:scenario_proof_2}
\end{split}
\end{equation}
We will motivate our choice of $1 - \frac{\beta}{N}$ in \cref{eq:scenario_proof_2} below.
To proof \cref{lemma:prob_lower_bound}, we need to show that for some $\discardedSet$, it holds that $\region_j(\lambda^\star_\discardedSet) \subseteq \region_j$.
However, we do not know \emph{a priori} (i.e., before observing the samples) for which set $\discardedSet$ this claim holds.
Specifically, there are $N$ possible sets $\discardedSet$ that we may consider, ranging from sizes $\vert \discardedSet \vert \in \{1,\ldots,N\}$.\footnote{Note that the case $\vert \discardedSet \vert = 0$ (i.e., no samples are contained in $\region_j$ at all) is treated as a special case in \cref{lemma:prob_lower_bound}.}
Let us, for each value of $\vert \discardedSet \vert$, denote by $\mathcal{A}_{\vert \discardedSet \vert}$ the event that
\begin{equation*}
    \amsmathbb{P} 
    \{
    \noise_k \in \Omega \colon
    \boxSample_{i \ell} \subseteq \region_j(\lambda^\star_\discardedSet) \} \geq \munderbar{p}.
\end{equation*}
From \cref{eq:scenario_proof_2}, we have that $\amsmathbb{P}\{ \mathcal{A}_{\vert \discardedSet \vert} \} = 1-\frac{\beta}{N}$, while for its complement $\mathcal{A}'_{\vert \discardedSet \vert}$ we obtain $\amsmathbb{P}\{ \mathcal{A}'_{\vert \discardedSet \vert} \} = \frac{\beta}{N}$.
Via Boole's inequality (the union bound), we know that
\begin{equation}
    \amsmathbb{P}^N \Big\{ \bigcap_{\xi = 1}^{N} \mathcal{A}_\xi \Big\} 
    = 
    1 - \amsmathbb{P}^N \Big\{ \bigcup_{\xi = 1}^{N} \mathcal{A}'_\xi \Big\}
    \geq 1 - \beta.
    \label{eq:UnionBound2}
\end{equation}
After observing the samples at hand, we determine $\discardedSet$ as per \cref{eq:set_R_lower_bound}, giving an expression in the form of \cref{eq:scenario_proof_2} for that set $\discardedSet$.
The probability that this expression holds cannot be smaller than that of the intersection of all events in \cref{eq:UnionBound2}. 
Thus, we obtain that
\begin{equation}
    \amsmathbb{P}^N \Big\{
        \amsmathbb{P} \{ 
                \noise_k \in \Omega \colon \boxSample_{i\ell} \subseteq \region_j
            \} \geq \munderbar{p}
    \Big\} \geq 1 - \beta,
    \label{eq:Proof6_lowerBound}
\end{equation}
where $\munderbar{p} = 0$ if $\discardedSet = \emptyset$, i.e., $\vert \discardedSet \vert = 0$ (which trivially holds with probability one), and otherwise, $\munderbar{p}$ is the solution to \cref{eq:scenario_proof_2} for that cardinality $\vert \discardedSet \vert$:
\begin{equation}
    \frac{\beta}{N} = \sum_{i=0}^{N - \vert \discardedSet \vert} \binom Ni (1-\munderbar{p})^i \munderbar{p}^{N-i},
\end{equation}
which is equivalent to \cref{eq:prob_lower_bound,eq:prob_lower_bound2}. 
Thus, we conclude the proof.

\begin{figure}[t!]
	\centering
	\includegraphics[width = .75\columnwidth]{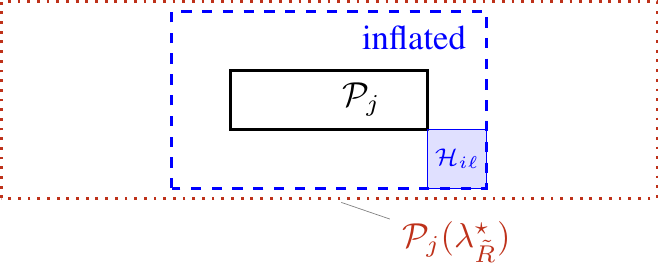}
	\caption{
	    Visualization of \cref{eq:inflated_region}, where \emph{inflated} represented the region $\region_j$ inflated by the size of the samples $\boxSample_{i\ell}$ (overapproximated by a box).
	}
	\label{fig:inflated_region}
\end{figure}

\subsection*{Proof of \cref{lemma:prob_upper_bound} (Upper Bound Probability)}
Before proving \cref{lemma:prob_upper_bound}, we discuss the claim from \cref{sec:PAC_intervals} that using scenario optimization may lead to conservative estimates of the upper bound transition probabilities.
To use scenario optimization for computing an upper bound probability, we must (analogous to computing a lower bound) find a subset of samples $\tilde{\discardedSet} \subset \{1,\ldots,N\}$ such that the solution to the scenario problem in \cref{eq:scenario_problem} yields a region $\region_j(\lambda^\star_{\tilde\discardedSet})$ that is (roughly speaking) \emph{at least as big as the region $\region_j$, inflated by the width of the sample sets $\boxSample_{i\ell}$}.
The intuition is that, to find a valid upper bound transition probability, we must find a region $\region_j(\lambda^\star_{\tilde\discardedSet})$ such that if $\boxSample_{i\ell} \not\subseteq \region_j(\lambda^\star_{\tilde\discardedSet})$, we always have $\boxSample_{i\ell} \cup \region_j = \emptyset$.
Formally, this condition is written as follows:
\begin{equation}
    \region_j(\lambda^\star_{\tilde\discardedSet}) \supseteq \{ x \in \R^n \colon \lVert x - y \rVert_2 \leq w, \, y \in \region_j \},
    \label{eq:inflated_region}
\end{equation}
where $w > 0$ is the (maximum) width of the set $\boxSample_{i\ell}$, as also shown in \cref{fig:inflated_region}.
In other words, the set $\region_j(\lambda^\star_{\tilde\discardedSet})$ for the solution $\lambda^\star_{\tilde{\discardedSet}}$ must be a \emph{superset} of the region $\region_j$, \emph{inflated} by (at least) the size of the set $\boxSample_{i\ell}$.

Depending on the shape of the region $\region_j$, it can be difficult to satisfy \cref{eq:inflated_region}.
In particular, if the region $\region_j$ has a narrow shape (as is the case in \cref{fig:inflated_region}), we must include many more samples in $\tilde{\discardedSet}$ than needed to find the lower bound using \cref{lemma:prob_lower_bound}.
As a result, using scenario optimization to compute upper bound probabilities would typically lead to very conservative results.

\paragraph{Proof of \cref{lemma:prob_upper_bound}.}
Assume we are given $N$ samples of a Bernoulli random variable with unknown probability $p$, and with sample sum denoted by $\tilde{\discardedSet}$.
For this Bernoulli random variable, Hoeffding's inequality~\cite{Boucheron2013concentrationInequalities} is traditionally stated in the following form:
\begin{equation}
    \amsmathbb{P}^N \Big\{ 
        pN \leq \tilde{\discardedSet} + \varepsilon N
        \Big\} \geq 1 - e^{-2\varepsilon^2 N},
        \label{eq:Hoeffding_proof1}
\end{equation}
for some $\varepsilon > 0$.
Thus, the expected value $pN$ over $N$ samples is upper bounded by the sample sum $\tilde{\discardedSet}$ plus the value of $\varepsilon$.
We are interested in the unknown probability $p$, instead of the sum over $N$ samples, so we rewrite \cref{eq:Hoeffding_proof1} as
\begin{equation}
    \amsmathbb{P}^N \Big\{ 
        p \leq \frac{\tilde{\discardedSet}}{N} + \varepsilon
        \Big\} \geq 1 - e^{-2\varepsilon^2 N}.
        \label{eq:Hoeffding_proof2}
\end{equation}
Moreover, let $\beta = e^{-2\varepsilon^2 N}$, and rewrite \cref{eq:Hoeffding_proof2} as
\begin{equation}
    \amsmathbb{P}^N \Big\{ 
        p \leq \frac{\tilde{\discardedSet}}{N} + \sqrt{\frac{1}{2N} \log(\frac{1}{\beta})}
        \Big\} \geq 1 - \beta.
        \label{eq:Hoeffding_proof3}
\end{equation}
In \cref{lemma:prob_upper_bound}, the unknown probability $p$ is the probability for a random sample $\boxSample_{i\ell}$ to be contained in the region $\region_i$:
\begin{equation*}
    p = \amsmathbb{P} \{ 
                \noise_k \in \Omega 
                \colon
                \boxSample_{i\ell} \cap \region_j \neq \emptyset
            \},
    \label{eq:Hoeffding_proof4}
\end{equation*}
and its sum $\tilde{\discardedSet}$ over $N$ samples is defined as in \cref{eq:discardedSet}.
Note that probability $p$ in \cref{eq:Hoeffding_proof3} cannot exceed 1, so we obtain
\begin{equation*}
    \amsmathbb{P}^N \left\{ 
        p \leq 
        \min \left\{ 1, \enskip \frac{\tilde{\discardedSet}}{N} + \sqrt{\frac{1}{2N} \log\Big(\frac{1}{\beta}\Big)} \right\}
        \right\} \geq 1 - \beta,
    \label{eq:Hoeffding_proof5}
\end{equation*}
which is the desired expression in \cref{eq:prob_upper_bound}.
This concludes the proof of \cref{lemma:prob_upper_bound}.

\subsection*{Proof of \cref{theorem:correctness} (Correctness of the iMDP)}
First, note that any transition prescribed by the optimal policy is also realizable on the original system, through the controller defined by \cref{eq:ControlLaw_opt}.
If the true transition probabilities $P(s_i,a_\ell)(s_j)$, for all $s_i, s_j \in \States$, $a_\ell \in \Actions$, are contained in their intervals, then it holds that
\begin{equation}
    V(\state_0, \alpha, c) \geq \reachProbDiscrRobustStar.
\end{equation}
Now, recall from \cref{eq:transition_probability} that transition probabilities are independent of the state in which an action is taken, i.e.,
\begin{align*}
    \transfuncImdp(s, a_\ell,s_j) &= \transfuncLow(s', a_\ell,s_j) \,\,\forall s, s', s_j \in \States, a_\ell \in \Actions.
\end{align*}
Thus, there at most $\vert \States \vert \cdot \vert \Actions \vert$ distinct transition probability intervals.
Each of these intervals contains its true probability with at least a probability of $1 - 2\tilde\beta N$.
Via Boole's inequality (the union bound), we know that all intervals are simultaneously correct with at least a probability of $1 - 2\beta N \cdot \vert\States\vert \cdot \vert\Actions\vert$, which concludes the proof.

\section{Approximate Sample Counting}
\label{sec:approximate_sample_counting}
Recall from \cref{sec:Abstractions} that, to compute \gls{PAC} probability intervals, we need to determine the counts of samples $\discardedSet$ and $\tilde{\discardedSet}$.
This amounts to counting, for every possible successor state $s_j$, the number of samples $\boxSample_{i\ell}^\iota$, $\iota = 1,\ldots,N$ that are contained in $\region_j$ (to determine $\discardedSet$), and the number having a nonempty intersection with $\region_j$ (to determine $\tilde{\discardedSet}$).
In this section, we introduce an approach to reduce the complexity of this procedure, especially when the value of $N$ is large.

\paragraph{Merging samples.}
Intuitively, the idea is to merge samples that are very similar, and to overapproximate these samples as a single, larger sample.
Formally, let $\rho > 0$ be a tuning parameter that reflects the maximum distance for two samples to be merged.
We merge two samples $\boxSample_{i\ell}^{(a)}$ and $\boxSample_{i\ell}^{(b)}$ if their centers, denoted by $h^{(a)} \in \boxSample_{i\ell}^{(a)}$ and $h^{(b)} \in \boxSample_{i\ell}^{(b)}$ are at most a $\rho$-distance apart, i.e.,
\begin{equation}
    \lVert h^{(a)} - h^{(b)} \rVert_2 \leq \rho.
    \label{eq:merging_samples}
\end{equation}
If \cref{eq:merging_samples} holds, we define one larger set $\boxSample_{i\ell}^{(a,b)} \supseteq \boxSample_{i\ell}^{(a)} \cap \boxSample_{i\ell}^{(b)}$ (without loss of generality, we define $\boxSample_{i\ell}^{(a,b)}$ as a hyperrectangle for simplicity).
Then, to determine sets $\discardedSet$ and $\tilde{\discardedSet}$ using \cref{eq:set_R_lower_bound} and \cref{eq:discardedSet}, respectively, a merged sample set is associated with \emph{the number of samples that it represents}.
For example, if $\boxSample_{i\ell}^{(a,b)} \subseteq \region_j$ (i.e., the merged sample that represents $\boxSample_{i\ell}^{(a)}$ and $\boxSample_{i\ell}^{(b)}$ is contained in $\region_j$), we add $2$ to the value of $\discardedSet$.
As a result, this procedure yields \emph{slightly more conservative} (depending on the value of $\rho$), \emph{yet sound estimates} of the counts in $\discardedSet$ and $\tilde{\discardedSet}$, and thus also of the \gls{PAC} probability intervals.

\paragraph{Tractable algorithm.}
Determining the best way to merge samples, however, is a problem of combinatorial complexity.
In our implementation, we thus use a heuristic in which we randomly select a sample, denoted by $\boxSample_{i\ell}^{(a)}$, that has not been merged yet.
We merge this sample with all other (non-merged) samples for which \cref{eq:merging_samples} holds, and we remove these samples from the list of non-merged samples.
To ensure termination, we mark $\boxSample_{i\ell}^{(a)}$ as a merged sample, even if no samples are within a $\rho$-distance of $\boxSample_{i\ell}^{(a)}$.
We repeat this procedure until no non-merged samples remain.

\paragraph{Reduction in complexity.}
While the improvement in computational complexity strongly depends on the model at hand, we have observed significant improvements in the experiments in \cref{sec:Experiments}.
For example, for the numerical experiments in \cref{sec:Experiments}, we used $20\,000$ samples to compute probability intervals, but using the proposed merging procedure with $\rho = 0.01$, we reduced this to around $1\,000$ merged samples.
    \section{Details on Numerical Experiments}
\label{sec:details_experiments}

\subsection*{Longitudinal Drone Dynamics}
To show that our approach is applicable to models with multiple uncertain parameters, we extend the longitudinal drone dynamics from \cref{ex:Model} with an uncertain spring coefficient $\zeta \in \R$, yielding the following model:
\begin{equation*}
        \state_{k+1} = 
        \begin{bmatrix}
            p_{k+1} \\ v_{k+1}
        \end{bmatrix} = 
        \begin{bmatrix}
            1 & \tau \\ -\frac{\zeta}{m} & 1 - \frac{0.1 \tau}{m}
        \end{bmatrix} \state_k +
        \begin{bmatrix}
            \frac{\tau^2}{2m} \\ \frac{\tau}{m}
        \end{bmatrix} \control_k + 
        \noise_k.
\end{equation*}
To write this model in the form of \cref{eq:uncertain_LTI_model}, we need four matrices $A_1, \ldots, A_4$ and $B_1, \ldots, B_4$, which are defined for the combinations of the minimum/maximum mass and spring coefficient.
For this experiment, we constrain the mass in the interval $0.9 \leq m \leq 1.1$ and the spring coefficient in $0.4 \leq \zeta \leq 0.6$.
We fix their nominal values as $\hat{m} = 1$ and $\hat\zeta = 0.5$.
We consider the same reach-avoid problem as in \cref{sec:Experiments} and we use the same partition into $480$ regions.
Moreover, we use $20$K samples to compute transition probability intervals of the \gls{iMDP}, with the approximate sample counting to reduce the computational complexity.

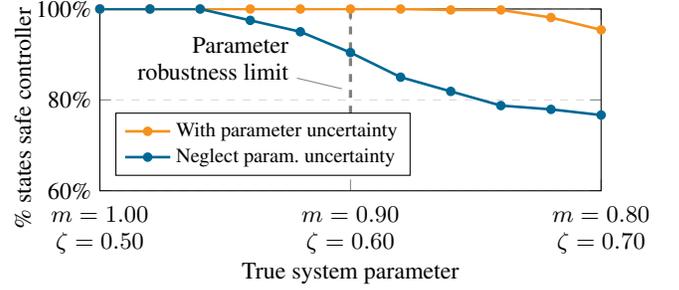
\begin{figure}[t!]
    \centering
        \footnotesize\begin{tikzpicture}
      \begin{axis}[
          width=0.98\linewidth,
          height=4cm,
          ymajorgrids,
          grid style={dashed,gray!30},
          xlabel={True system parameter},
          ylabel={\% states safe controller},
          x label style={at={(axis description cs:0.5,-0.35)}},
          y label style={at={(axis description cs:-0.12,0.40)}},
          x tick label style={align=center, yshift=-0.1cm},
          xmin = 0,
          xmax = 2,
          ymin=0.60,
          ymax=1,
          xtick={0,1,2},
          xticklabels={$m=1.00$\\$\zeta=0.50$,
                       $m=0.90$\\$\zeta=0.60$,
                       $m=0.80$\\$\zeta=0.70$},
          ytick={0.6,0.8,1.0},
          yticklabels={60\%,80\%,100\%},
          every axis plot/.append style={line width=1pt},
          legend cell align={left},
          legend columns=1,
          legend style={at={(0.62,0.43)},
                        nodes={scale=0.85, transform shape},
                        anchor=north east, 
                        column sep=0ex,}
        ]
        
        \addplot[mark=*, mark size=1.3pt, color=BurntOrange] table[x=x, y=y, col sep=semicolon] {Figures/Results/Oscillator/mass_spring/fraction_safe_parametric=True.csv};
        
        
        

        
        \addplot[mark=*, mark size=1.3pt, color=MidnightBlue] table[x=x, y=y, col sep=semicolon] {Figures/Results/Oscillator/mass_spring/fraction_safe_parametric=False.csv};
        
        
        
        
        \draw [dashed, very thick, gray] (axis cs:1,0.7) -- (axis cs:1,1.0);
        
        \node[pin={[pin distance=0.6cm, align=right]177.8:{Parameter\\robustness limit}}] at (axis cs:1,0.82) {};
        
        \legend{{With parameter uncertainty}, {Neglect param. uncertainty}}
      \end{axis}
    \end{tikzpicture}
    \vspace{-2em}
    \caption{Longitudinal drone dynamics of \cref{sec:Experiments}, extended by an uncertain friction coefficient. 
    These results confirm those of \cref{fig:harmonic_oscillator_plot} that our approach is 100\% safe up to the parameter robustness limit, while the approach that neglects epistemic uncertainty is not.}
    \label{fig:harmonic_oscillator_plot_2param}
\end{figure}

\paragraph{Results.}
We present a plot analogous to \cref{fig:harmonic_oscillator_plot} for the model with two uncertain parameters in \cref{fig:harmonic_oscillator_plot_2param}.
The figure shows the percentage of states with a safe controller (see \cref{sec:Experiments} for a definition), versus the values of the drone's mass $m$ and friction coefficient $\zeta$.
These results confirm those presented in \cref{sec:Experiments}, namely that our approach yields 100\% safe controllers up to the parameter robustness limit, while the baseline that neglects epistemic uncertainty is not safe.

\subsection*{Building Temperature Control}
Each room $i=1,\ldots,5$ of the building is modeled by its (zone) temperature $T_i^z \in \R$ and radiator temperature $T_i^r \in \R$.
Each room has a scalar control input $T^\text{ac}_i \in \R$ that reflects the air conditioning (ac) temperature, which is constrained within $15 \leq T^\text{ac}_i \leq 30$.
The change in the temperature of zone $i$ depends on the temperatures in the set of neighboring rooms, denoted by $\mathcal{J} \subseteq \{ 1,\ldots,5 \} \setminus \{i\}$.
Thus, the thermodynamics of the room temperature $T_i^z$ and radiator temperature $T_i^r$ of room $i$ are written as
\begin{align*}
    \dot{T}^z_i &= \frac{1}{C_i} \Big[ \sum_{j \in \mathcal{J}} \frac{T^z_j - T^z_i}{R_{i,j}} + \frac{T_\text{wall} - T^z_i}{R_{i,\text{wall}}} 
    \\
    & \quad\quad + m C_{pa} (T^\text{ac}_i - T^z_i) + P_i (T_i^r - T_i^z) \Big]
    \nonumber
    \\
    \dot{T}_i^r &= k_1 (T_i^z - T_i^r) + k_0 w (T_i^\text{boil} - T_i^r),
\end{align*}
where $C_i$ is the thermal capacitance of zone $i$, $R_{i,j}$ is the resistance between zones $i$ and $j$, $T_\text{wall}$ is the wall temperature, $m$ is the air mass flow, $C_{pa}$ is the specific heat capacity of air, and $P_i$ is the rated output of radiator $i$.
Moreover, $k_0$ and $k_1$ are constants, and $w$ is the water mass flow from the boiler.
We refer to our codes, provided in the supplementary material, for the precise parameter values used.

\paragraph{Modeling epistemic uncertainty.}
Important for the discussion here is that we assume the rated output of each radiator $i = 1,\ldots,5$ to be uncertain, within an interval of $0.8 \leq P_i \leq 1.2$.
We fix its nominal value to be $\hat{P}_i = 1$, so the uncertainty is $\pm 20\%$ around the nominal value.

\paragraph{Interactions between rooms as nondeterminism.}
Since directly discretizing the $10$D state space is infeasible, we use the procedure from \cref{sec:Algorithm} to capture \emph{any possible thermodynamic interaction} between rooms in the additive parameter $\disturbance_k \in \mathcal{Q}$ that belongs to the convex set $\mathcal{Q}$.
Specifically, the set $\mathcal{Q}_i$ affecting the thermodynamics of room $i \in \{1,\ldots,5\}$ is defined as follows (recall that $\SafeSet$ denotes the safe set):
\begin{equation}
    \mathcal{Q}_i = \big\{ \sum_{j \in \mathcal{J}} \frac{T^z_j - T^z_i}{R_{i,j}} \,\colon\, T^z \in \SafeSet \big\},
    \label{eq:BAS_disturbance}
\end{equation}
In other words, the uncertainty set $\mathcal{Q}$ is characterized by the maximal difference between $T_j^z$ and $T_i^z$ within the safe set, for all $j \in \mathcal{J}$, which is $\SI{5}{\celsius}$ for this specific reach-avoid problem (which was defined in \cref{sec:Experiments}).
Thus, depending on the other parameters, we can easily derive a set-bounded representation of $\mathcal{Q}$.

\paragraph{Discretization.}
We discretize the thermodynamics of a single room $i$ by a forward Euler method at a time resolution of $\SI{20}{\minute}$.
Moreover, we consider an additive Gaussian process noise $\noise_k$ on the room temperature of distribution $\Gauss[0]{0.002}$, and on the radiator temperature of distribution $\Gauss[0]{0.01}$.
As the model for room $i$ has only one uncertain parameter (the radiator power output $P_i$), we obtain a model in the form of \cref{eq:uncertain_LTI_model} with $r=2$ matrices $A_1, A_2$ and $B_1, B_2$ (we omit the explicit matrices for brevity).

\paragraph{Results.}
As described in \cref{sec:Experiments}, we apply our method with different partition sizes, and we compare two cases: 1) with the epistemic uncertainty, and 2) without the epistemic uncertainty, in which case we assume that the rated power output of each radiator is $P_i = \hat{P}_i = 1$.
We present the sizes of the obtained \glspl{iMDP} (the number of transitions is the number of $\transfuncImdp(s,a,s')$ in \cref{def:iMDP} with nonzero probability) and the run times in \cref{tab:temperature_control}.
We refer to \cref{sec:Experiments} for a more elaborate discussion of these results.

\begin{table}[t!]

\centering
\caption{\gls{iMDP} sizes and run times for different partitions on the temperature control problem (considering one room, decoupled from the others using the procedure from \cref{sec:Algorithm}).}

{\setlength{\tabcolsep}{4pt}

\scalebox{0.9}{
\begin{tabular}{cllll}
    \hline
    Epist.unc. & Partition & States & Transitions & Run time [s] 
    \\
    \hline
    & $15 \times 25$ & 378 & 84\,539 & 5.64
    \\
    & $25 \times 35$ & 878 & 308\,845 & 10.47
    \\
    & $35 \times 45$ & 1578 & 2\,103\,986 & 25.05
    \\
    & $50 \times 70$ & 3503 & 6\,149\,432 & 62.50
    \\
    & $70 \times 100$ & 7003 & 51\,742\,285 & 308.08
    \\
    \faIcon{check} & $15 \times 25$ & 378 & 100\,994 & 5.68
    \\
    \faIcon{check} & $25 \times 35$ & 878 & 463\,173 & 11.49
    \\
    \faIcon{check} & $35 \times 45$ & 1578 & 2\,932\,224 & 30.06
    \\
    \faIcon{check} & $50 \times 70$ & 3505 & 9\,520\,698 & 80.49
    \\
    \faIcon{check} & $70 \times 100$ & 7003 & 81\,763\,143 & 475.35
    \\
    \hline
\end{tabular}
}
}

\label{tab:temperature_control}
\end{table}

\begin{figure}[t!]
    \centering
    \includegraphics[width = \columnwidth]{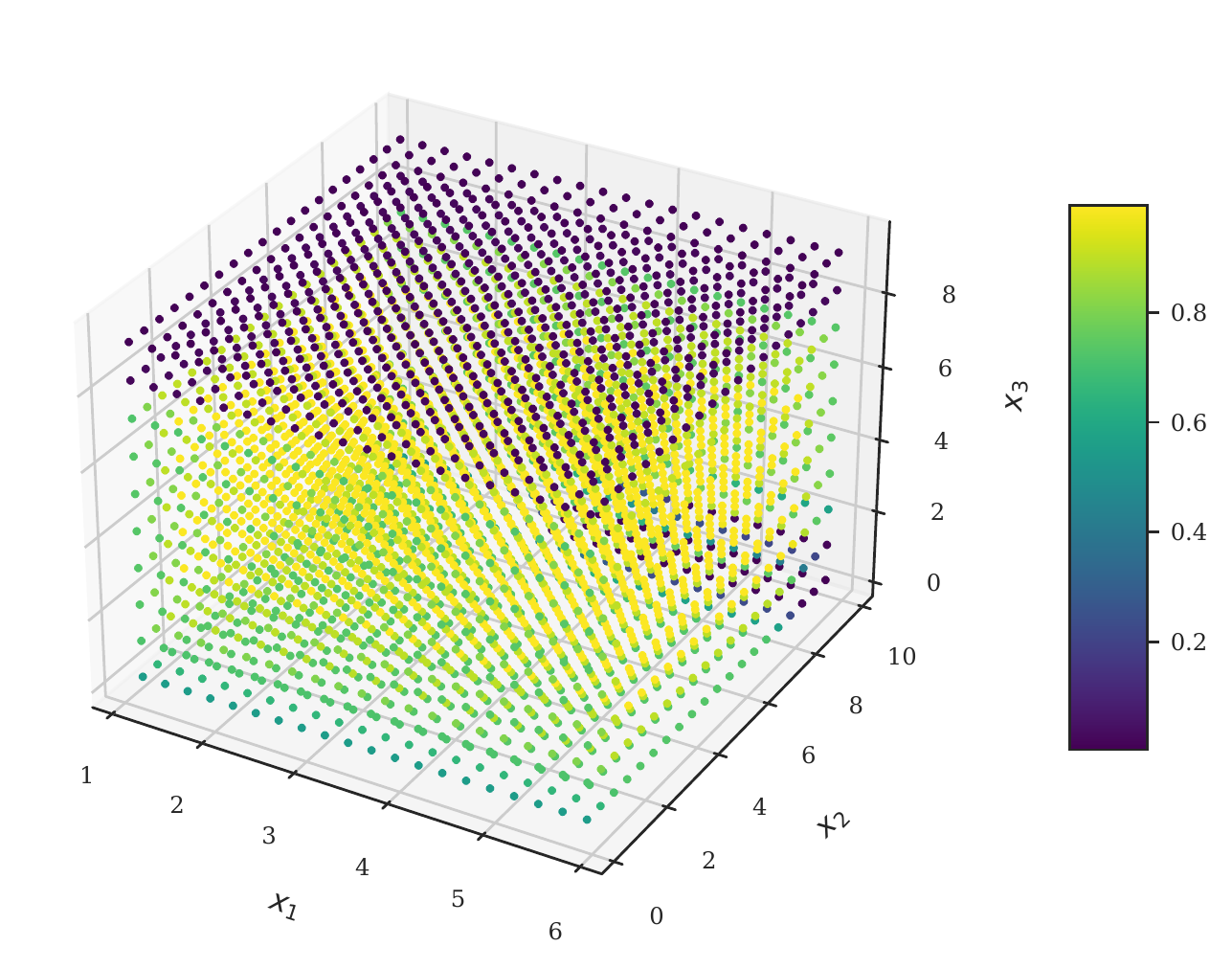}
    \caption{Maximum lower bound probabilities in the anesthesia delivery problem (for a grid of initial state) to remain in the safe set within the horizon of $20$ time steps.}
    \label{fig:drug_results}
\end{figure}

\subsection*{Automated Anesthesia Delivery}
We extend the automated anesthesia (propofol) delivery benchmark from~\citet{DBLP:conf/adhs/AbateBCHHLOSSVV18} with epistemic uncertainty in the pharmacokinetic system parameters.
This benchmark models the concentration of propofol administered to a patient as a \emph{three-compartment pharmacokinetic model}.
The continuous-time, parametric version of this model is written as follows:
\begin{equation}
    \dot{x} = 
    \begin{bmatrix}
        -(k_{10} + k_{12} + k{13}) & k_{12} & k_{13} \\
        k_{21} & -k_{21} & 0 \\
        k_{31} & 0 & -k_{31}
    \end{bmatrix} x + 
    \begin{bmatrix}
        \frac{1}{V_1} \\ 0 \\ 0
    \end{bmatrix} u,
    \label{eq:anaesthesia_model}
\end{equation}
where the state $x \in \R^3$ represents the propofol concentration in each of the three compartments, and where $u \in \R$ is the amount of propofol given to the patient.
Parameters $V_1$, $k_{ij}, i,j \in \{1,2,3\}$ are patient-specific parameters.
We discretize \cref{eq:anaesthesia_model} at a time step of $\SI{20}{\second}$, assuming a zero-order hold (i.e., piece-wise constant) control input.

\paragraph{Capturing uncertainty.}
We use the same parameter values as those reported by~\citet{DBLP:conf/adhs/AbateBCHHLOSSVV18} and assume that parameters $k_{10}$, $k_{21}$, and $V_1$ are uncertain within $\pm 10\%$ of their nominal values.
To capture these uncertain parameters, we write \cref{eq:anaesthesia_model} in the form of \cref{eq:uncertain_LTI_model} with $r = 2^3 = 8$ matrices $A_1,\ldots,A_r$ and $B_1,\ldots,B_r$.
In addition to this epistemic uncertainty, we also add a stochastic process noise, which is assumed to have a zero-mean Gaussian distribution with a diagonal covariance matrix $10^{-3} I_3$, where $I_3 \in \R^{n\times n}$ is the identity matrix.

\paragraph{Planning problem and abstraction.}
We consider a safety problem, where the goal is to keep the propofol concentration within a safe set $\SafeSet = [1, 6] \times [0, 10] \times [0, 10]$ for $20$ discrete time steps.
We partition the continuous state space into $20 \times 20 \times 20 = 8\,000$ discrete regions, yielding an \gls{iMDP} with this same number of states.
We apply \cref{theorem:probability_interval} with $20$K noise samples, but with the approximate counting scheme described in \cref{sec:approximate_sample_counting}, we reduce the number of samples to around $1\,000$.

\paragraph{Results.}
\change{We present a 3-dimensional heatmap of the optimal reachability probabilities $\reachProbDiscrRobustStar$ under the robust optimal \gls{iMDP} policy for different initial states in \cref{fig:drug_results}.}
Recall from \cref{theorem:correctness} that these results are lower bound guarantees on the performance of a controller in practice.
We observe that, except when the initial concentration in compartment 3 is too high (approximately above 8), we are able to synthesize a controller that enables the system to remain in the safe set for the horizon of $20$ discrete steps.
However, when the initial concentration in compartment 3 is too high, no safe controller could be synthesized, as reflected by the low probabilities shown in \cref{fig:drug_results}. 
\fi

\end{document}